\documentclass[sigconf]{acmart}
\renewcommand\footnotetextcopyrightpermission[1]{} 

\copyrightyear{2017} 
\acmYear{2017} 
\setcopyright{acmlicensed}
\acmConference{CIKM'17}{}{November 6--10, 2017, Singapore.}
\acmPrice{15.00}
\acmDOI{https://doi.org/10.1145/3132847.3133042}
\acmISBN{ISBN 978-1-4503-4918-5/17/11} 

\usepackage{booktabs}
\usepackage{epsfig}
\usepackage{graphicx,xspace,comment,subfigure,algpseudocode,multirow,url,amsmath,enumitem,amsthm}
\usepackage[ruled,vlined,oldcommands,linesnumbered]{algorithm2e}
\usepackage{lipsum,graphicx,multicol,vwcol}
\usepackage{tabularx}
\usepackage{setspace}

\newcommand{\be}{\begin{equation}}
\newcommand{\ee}{\end{equation}}
\newcommand{\bea}{\begin{eqnarray}}
\newcommand{\eea}{\end{eqnarray}}

\newcommand{\eol}{\end{enumerate}\setlength{\itemsep}{-\parsep}}

\newcommand{\parab}[1]{\paraspace\noindent{\textbf{#1}} }
\newcommand{\paraspace}{\vspace{0.05in}}

\def\hh{\textsc{HashHeap}}

\relax

\let\hat\widehat

\def\Nl{{\mathbb N}}

\def\E{{\mathbb E}}
\def\Pr{{\mathbb P}}

\def\cZ{{\mathcal Z}}

\def\var{\mathop{\mathrm{Var}}}

\def\argmin{\mathop{\mathrm{arg\,min}}}

\def\psh{\textsc{PBA}}
\def\papsh{\textsc{PBASH}}
\def\prash{\textsc{PBASH}}
\def\ash{\textsc{ASH}}

\algblock{ParFor}{EndParFor}
\algnewcommand\algorithmicparfor{\textbf{parallel for}}
\algnewcommand\algorithmicpardo{\textbf{do}}
\algnewcommand\algorithmicendparfor{\textbf{end parallel}}
\algrenewtext{ParFor}[1]{\algorithmicparfor\ #1\ \algorithmicpardo}
\algrenewtext{EndParFor}{\algorithmicendparfor}

\algnotext{EndFor}
\algnotext{EndIf}
\algnotext{EndWhile}
\algnotext{EndProcedure}
\algnotext{EndParFor}
\algnotext{EndFunction}

\begin{document}

\title{Stream Aggregation Through Order Sampling}
\thanks{This material is based in part upon work supported by the
 National Science Foundation under Grant Numbers CNS-1618030 and CNS-1701923.}

\author{Nick Duffield}

\affiliation{%
  \institution{Texas A\&M University}
  \streetaddress{College Station, TX}
}
\email{duffieldng@tamu.edu}

\author{Yunhong Xu}
\affiliation{
  \institution{Texas A\&M University}
  \streetaddress{College Station, TX}
}
\email{yunhong@tamu.edu}

\author{Liangzhen Xia}
\affiliation{
  \institution{Texas A\&M University}
  \streetaddress{College Station, TX}
}
\email{xialiangzhen123@tamu.edu}

\author{Nesreen K. Ahmed}
\affiliation{%
  \institution{Intel Labs}
  \streetaddress{Santa Clara, CA}
}
\email{nesreen.k.ahmed@intel.com}

\author{Minlan Yu}
\affiliation{%
  \institution{Yale University}
  \streetaddress{New Haven, CT}
}
\email{minlan.yu@yale.edu}

\begin{abstract}
This paper introduces a new single-pass reservoir weighted-sampling
stream aggregation algorithm, Priority-Based Aggregation (PBA). While order sampling is a powerful and efficient method for weighted sampling from a stream of uniquely keyed items, there is no current algorithm that realizes the benefits of order sampling in the context of stream aggregation over non-unique keys. A naive approach to order sample regardless of key then aggregate the results is hopelessly inefficient. In distinction, our proposed algorithm uses a single persistent random variable across the lifetime of each key in the cache, and maintains unbiased estimates of the key aggregates that can be queried at any point in the stream. The basic approach can be supplemented with a Sample and Hold pre-sampling stage with a sampling rate adaptation controlled by PBA. This approach represents a considerable reduction in computational complexity compared with the state of the art in adapting Sample and Hold to operate with a fixed cache size. Concerning statistical properties, we prove that PBA provides unbiased estimates of the true aggregates. We analyze the computational complexity of PBA and its variants, and provide a detailed evaluation of its accuracy on synthetic and trace data. Weighted relative error is reduced by 40\% to 65\% at sampling rates of 5\% to 17\%, relative to Adaptive Sample and Hold; there is also substantial improvement for rank queries.
\end{abstract}

\keywords{Priority Sampling; Aggregation; Subset Sums; Heavy Hitters}

\maketitle

\section{Introduction}

\subsection{Motivation}

We consider a data stream comprising a set of (key, value) pairs $(k_i,x_i)$. Exact aggregation would entail computing the total value $X_k = \sum_{i:k_i=k}x_i$ for each distinct key $k$ in the stream. For many applications, this is unfeasible due to the storage required to accommodate a large number of distinct keys. This constraint has motivated an extensive literature on computing summaries of data streams. Such summaries can be used to serve approximate queries concerning the aggregates through estimates $\hat X_k$ of $X_k$, typically accomplished by assigning resources to the more frequent keys. 

This problem of stream aggregation has drawn the attention of researchers in Algorithms, Data Mining, and Computer Networking, who have proposed a number of solutions that we review in Section~\ref{sec:related}. Nevertheless, applications of this problem continue to emerge in new settings that bring their own challenges and constraints. These include: streams of transactional data generated by user activity in Online Social Networks \cite{7852324}, transactional data from customer purchases in online retailers \cite{Lo:2016:UBL:2939672.2939729}, and streams of status reports from customer interfaces of utility service providers reported via domestic Internet service \cite{Liu:2016:SMD:3015779.3004295}. 

A well-established application for real-time streams of operational traffic measurements collected by Internet Service Providers (ISPs) has gathered renewed interest in the context of Software Defined Networks (SDN) \cite{Yu:2013:SDT:2482626.2482631}. These provide the opportunity to move beyond industry standard summaries based on Sampled NetFlow and  variants \cite{netflow}. Data Center operators increasingly wish to control traffic at a finer space and time granularity than has been typical for Wide Area Networks, requiring per flow packet aggregates over time scales of seconds or shorter \cite{roy2015inside,li2016flowradar}. 
An important goal is to balance traffic loads over multiple network paths, and between servers. Two distinct analysis functions can support this goal:
\begin{enumerate}[leftmargin=0pt,itemindent=15pt,itemsep=0pt,topsep=0pt,parsep=0pt,partopsep=0pt]
\item[$\bullet$] \textsl{Heavy Hitter Identification.}
Heavy Hitters (HHs) are flows or groups of flows that contain a disproportionate 
fraction of packets and/or bytes. These may be present in the exogenous loads, or may be indicative of underlying problems in the load balancing mechanisms \cite{fb2014dcfabric}.
\item[$\bullet$] \textsl{General Purpose Summarization.} (key, aggregate) summaries over flows or groups for flows can be further aggregated over arbitrary subpopulation selectors, e.g., for what-if analyses for load balancing. This aggregation capability is present in Stream Databases developed to run on high speed traffic measurement systems \cite{cranor2003}.
\end{enumerate}

Sampling is an attractive summarization method for supporting applications including those just described. First, sample sets can serve
downstream applications designed to work with the original data, albeit with approximate results. Second, sampling supports retrospective queries using selectors formulated after the summary was formed. This enables sum queries over subpopulations whose constituent keys are not individually heavy hitters.
Finally, sampling can often be tuned to meet specific goals constraints on memory, computation and accuracy that match data characteristics to query goals. We distinguish between two types of space constraint. The \textsl{working storage} used during the construction of the summary may be limited. An example is stream summarization of Internet traffic by routers and switches, where fast memory used to aggregate packet flows is relatively expensive \cite{Moshref:2015:SSR:2716281.2836099}. But the \textsl{final storage} used for the finished summary generally has a smaller per item requirement than the working storage. A final storage constraint can apply, for example, when storage must be planned or pre-allocated for the summary, or when the size of the summary is limited in order to bound the response time of subsequent queries against it. 

Reservoir Sampling \cite{Vitter:85} is commonly used to obtain a fixed size sample. In stream aggregation reservoir sampling, an arriving item $(k,x)$ is used to modify the current aggregate estimate $\hat X_k$ or $X_k$ if $k$ is in the reservoir, e.g., by adding $x$ to $\hat X_k$. If $k$ is not in the reservoir, and the capacity of $m$ is already used, a random decision is made whether to discard the arriving item, or to instantiate a new aggregate for $k$ while discarding one of the items currently in the reservoir.  In general, discard probabilities are not uniform, but are weighted as a function of aggregate size to realize estimation goals for subsequent analysis.  In addition, estimates of retained items must be adjusted in order to maintain statistical properties of the aggregate estimates, such as unbiasedness. The time complexity to process an arriving item and adjust the estimates of the retained items is a crucial determinant for the computational feasibility of stream aggregation. Fixed size summaries are essential in cases where the stream load can vary significantly over time and is not otherwise controlled. A prime example comes from Internet traffic measurement, where the offered load can varying significantly both due to time-of-day variation, and due to exogenous events such as routing changes. Reservoir sampling acts to adapt the sampling to variations in the rate of arriving items, e.g. to take a periodic fixed size sample per router interface.

Order sampling has been proposed as a mechanism to implement uniform and weighted reservoir sampling in the special case that items have unique keys \cite{hajek:1960}. In order sampling, all sampling decisions depend on a family of random order variables generated independently for each arriving item. For arrival at a full reservoir of capacity $m$, from the $m+1$ candidate items (those currently in the reservoir and the arriving item) the item of lowest order is discarded. Several order sampling schemes have been proposed to fulfill different weighted sampling objectives; including Probability Proportional to Size (PPS) sampling \cite{rosen1997a} also known as Priority Sampling \cite{DLT:jacm07}, and Weighted Sampling without Replacement \cite{Rosen1972:successive,ES:IPL2006,bottomk07:ds}.
Stream order sampling can be implemented as a priority queue in increasing order \cite{DLT:jacm07}. While order sampling can be applied directly to an unaggregated stream and samples aggregated post-sampling, this is clearly wasteful of resources.

\subsection{Contribution and Summary of Results}\label{sec:cont}

This paper proposes Priority-Based Aggregation (\psh), a new sampling-based algorithm for stream aggregation built upon order sampling that can provide unbiased estimates of the per key aggregates. \psh\ and its variants provide greater accuracy across a variety of heavy hitter and subpopulation queries than competitive methods in data driven evaluations. Our specific contributions are as follows:
\begin{trivlist}
\item  \textsl{Estimation Accuracy.} \psh\ is a weighted sampling algorithm developed from Priority Sampling that yields a stream summary in the form of unbiased estimates of all aggregates in the stream. A modification of \psh\ uses biased estimation to reduce error for smaller aggregates, while having a negligible impact on accuracy for larger aggregates. In experimental comparisons with a comparable sampling based method, Adaptive Sample and Hold \cite{Estan2002,Cohen:2015:SSF:2783258.2783279}, our methods reduced weighted relative estimation error over all keys by between 38\% and 65\% at sampling rates between 5\%  and 17\%  when applied to synthetic and network traffic traces. The accuracy for rank queries was also improved.
\item\textsl{Computational Complexity.}
To the best of our knowledge, \psh\ is the first algorithm to employ order sampling based on a single random variable per key in the context of stream aggregation.
This enables \psh\ to achieve low computational complexity for updates. It is average $O(1)$ to process each arrival that is either added to a current aggregate, or that presents a new key that is not selected for sampling. The exception comes when an arriving key not currently in storage replaces an existing key; the complexity of this step is worst case $O(\log m)$ in a reservoir of capacity $m$. Retrieval of the estimates is $O(1)$ per key.
\item\textsl{Priority-Based Adaptive Sample and Hold} (\papsh). We incorporate the well known weighted Sample and Hold \cite{Estan2002} algorithm as a pre-sampling stage, for which the sampling probabilities are controlled from the adaptation of the \psh\ second stage. This enables us to exploit the computational simplicity of the original (unadaptive) Sample and Hold algorithm while taking advantage of the relatively low computational adaptation costs of \psh, as compared with existing versions of Adaptive Sample and Hold \cite{KME:sigmetrics05,Cohen:2015:SSF:2783258.2783279}. 
\end{trivlist}

The outline of the rest of paper is as follows. In Section~\ref{sec:related} we review related work to give a more detailed motivation for our approach and set the scene for our later experimental evaluations. Section~\ref{sec:framework} describes the \psh\ algorithm and establishes unbiasedness of the corresponding estimators. 
Section~\ref{sec:optimization} describes four optimizations of these basic algorithms. Section~\ref{sec:deferred} describes \textsl{Deferred Update} for which we show that the unbiasing of estimates that must be performed on all aggregates after another is discarded can be deferred, for each such aggregate until an item with matching key arrives. Section~\ref{sec:preagg} describes pre-aggregation of successive items with the same key in the input stream. Section~\ref{sec:pash} describes the use of Sample and Hold as an initial sampling stage, and how its adaptation is controlled from \psh. Section~\ref{sec:error} describes a scheme to reduce estimation errors for small aggregates through the introduction of bias. 
Section~\ref{sec:structure} specifies the algorithm incorporating these optimizations, describes our implementation, and reports on computational and space complexity. Section~\ref{sec:eval} describes data driven evaluations, before we conclude in Section~\ref{sec:conclude}. Proofs are deferred to Section~\ref{sec:proofs}.

\section{Related Work}\label{sec:related}

In the earliest work in reservoir sampling $k$ items from a stream of distinct keys \cite{Vitter:85}, the $n^{\textrm th}$ item is chosen with probability $1/n$, giving rise to a uniform sample. To approximately count occurrences in a stream with repeated keys, Concise Samples \cite{GM:sigmod98} used uniform sampling, maintaining a count of sampled keys. In network measurement Sampled NetFlow \cite{netflow} takes a similar approach maintaining an aggregate of weights rather than counts. In Counting Samples \cite{GM:sigmod98}, previously unsampled keys are sampled with a certain probability, and if selected, all matching keys increment the key counter with probability 1. Sample and Hold \cite{Estan2002} is a weighted version of the same approach. 
Both schemes can be extended to adapt to  a fixed cache size, by decreasing the sampling probability and resampling all current items until one or more is ejected.
The set of keys cached by \ash\ is a PPSWR sample (sampling probability proportional to size without replacement), also known as bottom-k (order) sampling with exponentially distributed ranks \cite{bottomk07:ds,bottomk:VLDB2008,Rosen1972:successive}.
The comparisons of this paper use the form of ASH for Frequency Cap Statistics from \cite{Cohen:2015:SSF:2783258.2783279}, applied in the case of unbounded cap; see also an equivalent form in \cite{cohen2012don}. The number of deletion steps from a reservoir of size $m$ in  a stream of length $n$  is $O(n \log m)$ and each such deletion step must process $O(m)$ items, based on generation of new randomizers variables for each item. By contrast, \psh\ requires only a single randomizer per key, and is able to maintain items in a priority queue from which discard cost in only $O(\log m)$. Concerning memory usage, \psh\ requires maintenance of larger working storage per item, while the implementation of \ash\ in \cite{Cohen:2015:SSF:2783258.2783279} temporarily requires a similar amount during the discard step. Final storage requirements are the same. Step Sampling \cite{Cohen:2007:AEA:1298306.1298344} is a related approach in which intermediate aggregates are exported. 

Beyond sampling, many linear sketching approaches have been proposed; see e.g. \cite{AMS99,stable-sk,Johnson1986,Cormode04animproved,Gilbert:2002:FSA:509907.509966}.
More recently, $L_p$ methods  have been proposed in which each key is sampled with probability proportional to a power of its weight \cite{andoni,MoWo:SODA2010,Jowhari:Saglam:Tardos:11}. A general approach to sketch frequency statistics in a single pass is proposed in \cite{Braverman:2010:ZFL:1806689.1806729}, with applications to network measurement in \cite{Liu:2016:OSR:2934872.2934906}.
A drawback of sketch methods is that for a given accuracy, their space is logarithmic in the size of the key domain, which can be problematic for large domains such as IP addresses. Retrieval of the full set of aggregates (as opposed to query on specific keys) is costly, requiring enumerating the entire domain for each sketch; tuning of the sketch for specific queries, e.g., using dyadic ranges, is preferable. In our case, the full summary can be read directly in $O(m)$ time.
Space factors in the sketch-based methods also grow polynomially with the inverse of the bias, whereas our method enables unbiased estimation.
Beyond these comments, we do not perform an explicit comparison with sketch-based methods, instead referring the reader to a comparative evaluation of sketches with \ash\ for subpopulation queries in \cite{cohen2012don}.

Finally, weighted reservoir priority sampling from graph streams of unique edges has recently been developed in \cite{AhmedVLDB2017}, building on the conditionally independent edge sampling \cite{ahmed2014gsh}.

\section{Priority-Based Aggregation}\label{sec:framework}

\subsection{Preliminaries on Priority Sampling}
Priority Sampling $m$ items from a set of $n>m$ weights $\{x_i:\
i\in[n]\}$ is accomplished as follows. For each item $i$ generate
$u_i$ uniformly in $(0,1]$, and compute its priority
$r_i=x_i/u_i$. Retain the (random) top $m$ priority items, and for each such
item define the estimate $\hat x_i = \max\{x_i,z\}$, where $z$ is the
$(m+1)^{\textrm {st}}$ largest priority. For the remaining $n-m$ items
define $\hat x_i=0$. Then for each $i$, $\E[\hat x_i]=x_i$ where the
expectation is taken of the distribution of the $\{u_i:
i\in[n]\}$. Priority sampling can be implemented as reservoir streams sampling,
taking the first $m$ items, then processing the remaining $n-m$ items
in turn, provisionally adding each to the reservoir then using the above algorithm to discard one item.

\subsection{Algorithm Description}
We consider a stream of items $\{(k_t,x_t)\}_{t\in T }$ where
$T=[|T|]=\{1,2,\ldots,|T|\} \subset\Nl$.
$x_t>0$ is a \textbf{size} and $k$ a \textbf{key} that is a member
of some keyset $K$. Let
\be
X_{k,t}=\sum_{s\le t, k_s = k} x_s
\ee
denote the total size of items with key
$k$ arriving up
to time $t$ whose key is $k$. Let $K_t$ denote the set of unique
keys arriving up to and including time $t$. 
We aim to construct a fixed size random summary $\{\hat X_{k,t}: k\in \hat K_t\}$ where
$\hat K_t\subset K_t$ with $|\hat K_t|\le m$ which provides unbiased
estimates over \textsl{all} of $K_t$
$\E[\hat X_{k,t}]=X_{k,t}$ for all $k\in K_t$. Implicitly $\hat
X_{k,t}=0$ for $k\notin\hat K_t$.

To accomplish our goal we extend Priority Sampling to include
aggregation over repeated keys.
Sampling will be
controlled by a family of weights $\{W_{k,t}: k\in \hat
K_t\}$. These generalize the usual fixed weights of priority sampling in that
they can be both random and time dependent, although within certain
constraints that we will specify. The arrival $(k,x)=(k_t,x_t)$ is processed  as follows:
\begin{enumerate}[leftmargin=0pt,itemindent=15pt,itemsep=5pt]
\item
If the arriving key is in the reservoir, $k\in \hat K_{t-1}$ then we increase $X_{k,t}=X_{k,t-1}+x$,
leave the sample keyset  unchanged,
$\hat K_t = \hat K_{t-1}$,
and await the next arrival.
\item If the arriving key is not in the reservoir, $k\notin \hat K_{t-1}$, then we provisionally admit $k$ to
  the sample set forming $\hat K'_t=\hat K_{t-1}\cup\{k\}$. We initialize
$\hat X_{k,t}$ to $x$, $q_k$ to $1$, and generate 
the random $u_{k}$ uniformly on $(0,1]$. Then:
\begin{enumerate}[label=(\alph*),leftmargin=15pt,itemindent=15pt,topsep=5pt]
\item If $|\hat K|\le m$ we
set $\hat K_t = \hat K'_t$ and await the next arrival.
\item Otherwise $|\hat K|>m$, we discard the key
\[
 k^*=\textstyle{\argmin_{k'\in\hat  K_t}} W_{k',t}/u_{k'}
\]  from $\hat K'_t$ and set
$z^*=W_{k^*,t}/u_{k^*}$. For each remaining $k'\in \hat K$ set
$q_{k',t}=\min\{q_{k',t-1},W_{k,t}/z^*\}$ and $ \hat X_{k',t}=\hat
X_{k',t-1} q_{k,t-1}/q_{k',t}$.
\end{enumerate}
\end{enumerate}
While the description above is convenient for
mathematical analysis, we defer a formal specification to Section~\ref{sec:structure},
where Algorithms~\ref{alg4} and \ref{alg-psh-com} incorporate
optimizations described in Section~\ref{sec:optimization} that improve
performance relative to a literal implementation of steps (1), (2),
(2a), (2b) above.

\subsection{Unbiased Estimation} 

We now establish unbiasedness of $\hat X_{k,t}$ 
when $W_{k,t}$ is the cumulative increase in the size in $k$ 
since $k$ was last admitted to
the sample. 
For each key $k$ 
let $T_k$ denote the set of times $t$ at which $k$ was admitted to a
full reservoir, i.e., 
\be
T_k=\{t: k\ne \hat K_{t-1},\ k\in \hat K_t, |\hat K_{t-1}|=m\}
\ee
When $k\in\hat K_{t-1}$, let $\tau_{k,t}=\max[0,t-1]\cap T_k$
denote the most recent time prior to $t$ at which $k$ was admitted to the
reservoir, and for the arriving key $k_t$ we set $\tau_{k_t,t}=t$
prior to admission. 

Let $T^0=\{t: k_t\notin \hat K_{t-1}\}\subseteq T$ denote the times at
which the arriving key was not in the current sample. 
Let $\tau_t=\max [0,t-1]\cap T^0$ denote the most recent time prior to $t$ that
an arriving key was not the sample.  For an integer interval $Y$ we
will use the notation  $Y^0=T^0\cap Y$. 
For any $t\in T$ and $k\in \hat K'_t$, 
$u_k$ was generated at time $\tau_{k,t}$.
If $k$ is discarded from $\hat K'_t$, a
subsequent arrival of $k$ in an item will have a new independent $u_k$
generated. 

Our first version of \psh\ is governed by the exact weights $W_{k,t}$
that the total size in
key $k$ of arrivals since $k$ was most recently admitted to
sample, i.e.,
\be
W_{k,t}=X_{k,t}-X_{k,\tau_{k,t}-1}=\sum_{s\in[\tau_{k,t},t]:k_s=k}x_s
\ee
Note that $W_{k,t}$ can be maintained in the sample set by
accumulation. For each $t\in T^0$ and $i\in\hat K'_t$ let 
\be\label{eq:z}
z_{i,t}= \min_{j\in \hat K'_t\setminus\{i\}}\frac{W_{j,t}}{u_j}.
\ee
and $z_s$ denote the unrestricted minimum $z_s=\min_{j\in\hat K'_t}\frac{W_{j,t}}{u_j}$. 
The conditions under which $i\in\hat K'_t$ survives sampling are
\be\label{eq:set:iter}
\{i\in\hat K_t\}=\{i \in \hat K'_t\}\cap \{W_{i,t}/u_i > z_{i,t}\}\ee 
As a consequence $z_{i,s}=z_s$ if $i\in \hat K_s$.
For $t\in T^0$ define \be\label{eq:q} 
q_{k,t}= \min\{1,\min_{s\in [\tau_{k,t},t]^0}W_{k,s}/z_{s}\}
\ee
and 
\be\label{eq:Qdef}
Q_{k,t}=\left\{
\begin{array}{ll}
q_{k,t}& \mbox{if }k=k_t\\
q_{k,t}/q_{k,\tau_t},&\mbox{otherwise}\\
\end{array}
\right.
\ee
For $k\in K_t$, define $\hat X_{k,t}$ iteratively by
\be\label{eq:iter}
\hat X_{k,t}=\left\{
\begin{array}{ll}
\hat X_{k,t-1} + \delta_{k,k_t} x_t&\mbox{$t\notin T^0$}\\
(\hat X_{k,t-1} + \delta_{k,k_t} x_t)/Q_{k,t} & \mbox{$t\in T^0$,
                                                $k\in\hat K_t$}\\
0,&\mbox{otherwise}\\
\end{array}
\right.
\ee
where $\delta_{i,j}=1$ if $i=j$ and $0$ otherwise.
The proof of the unbiasedness of $\hat X_{k,t}$ is deferred to Section~\ref{sec:proofs}.
\begin{theorem}\label{thm:simple}
$\hat X_{k,t}$ is unbiased: $\E[\hat X_{k,t}]=X_{k,t}$.
\end{theorem}

We have also proved that replacing $W_{k,t}$ with an affine function
of the current estimator $\hat X_{k,t}$ also yields an unbiased
estimator at the next time slot. This has the utility of reducing
memory usage since a separate $W_{k,t}$ per aggregate is not
needed. However, we also found in experiments that this estimator was
not so accurate. For both variants of the estimator, we can derive
unbiased estimators of $\var(\hat X_{k,t})$, These can be used to
establish confidence intervals for the estimates. Due to space
limitations we omit further details on all the results summarized
in this paragraph.

\section{Optimizations}\label{sec:optimization}

\subsection{Deferred Update}\label{sec:deferred}

For each $i$, $q_{i,t}$ is computed as the minimum over $s$ of
$W_{i,s}/z_s$.  As it stands, this is more complex that the
corresponding computation in Priority Sampling for fixed weights
$W_i$, where $W_i/z^*_t$ is computed
once for each arrival.  By comparison, it appears that in principle,
we must update $q_{i,t}$ for all $i\in K_t$ at each $t\in T^0$. We now establish that for each key $k$, $q_{k,t}$ needs only be updated when an item with key $k$ arrival, i.e.., at $t$ for which $k_t=k$. Updates for times $t$ in $T^0$ for which $k_t\ne k_t$ can be deferred until the first time $t'>t$ for which $k_{t'}=k$, or whenever an estimate of $\hat X_{k,t}$ needs to be computed. This property is due to the constancy of the fixed weights between updates and the monotonicity of the sequence $z^*_t$.
For $t\in T^0$ let $z^*_t=\max_{s\in[0,t]^0}\{z_s\}$.

Let $d_t$ denote the key that is discarded from $\hat K'_{t-1}$ at time
$t\in T^0$, i.e., $\{d_t\}=\hat K'_{t-1}\setminus \hat K_t$. When $t\in T^0$ and 
$i\in \hat K_t$ define $q^*_{k,t}$ recursively by
\be
q^*_{i,t}=\min\{q^*_{i,\tau_t},W_{i,t}/z^*_{t}\}
\ee\label{eq:iter:pstar}
unless $k_t=i$ in which case $q^*_{i,t}=\min\{1,W_{i,t}/z^*_t\}$. The proof of the following result is detailed in Section~\ref{sec:proofs}.

\begin{theorem}\label{thm:delay-alt}
\begin{itemize}
\item[(i)] $t\in T$  implies $z^*_t=z_t$. 
\item[(ii)] $q_{i,t}=q^*_{i,t}$ for all $t$ where these are defined..
\end{itemize}
\end{theorem}

Theorem~\ref{thm:delay-alt} enables
computational speedup as compared with updating each key probability
at each $t\in T^0$.   Since $z^*_t$ is monotonic in $t$, we only need
to update the probabilities $q_{i,t}$ for links $i$ whose weight
increases after admitting a key at time $t$.  Likewise, we perform a final update
at the end of the stream, or at any intermediate time when an
estimate is required.

\subsection{Pre-aggregation}\label{sec:preagg}

Pre-aggregation entails summing weights over consecutive instances of
the same key before passing to \psh. Pre-aggregation saves on computational complexity of updating priorities, instead of updating a single counter. This also results in an unbiased estimator whose variance at least as large as \psh.

\subsection{Priority-Based Adaptive Sample and Hold}\label{sec:pash}

Sample and Hold \cite{Estan2002} with a fixed parameter is a simple
method to preferentially accumulate large aggregates. However, in this
form, Sample and Hold cannot adapt to variable load or a fixed buffer.
Adaptive Sample and Hold (ASH) \cite{Estan2002,KME:sigmetrics05}
using resamples to selectively discard from the reservoir.  We
propose to retain the advantages of Sample and Hold within an
adaptive framework by using it as a front end to \psh, with its sampling parameters adapted directly from the time-varying threshold of \psh.

We call this coupled system Priority-Based Adaptive Sample and Hold
(\papsh). When an arriving item $(k,x)$ finds its key $k$ is not in the
current sample $\hat K_t$, the item is sampled with probability
$p_t(x)=\min\{1,w/z_t^*\}$ where the current threshold
$z^*_t$  provides scale that takes into account
the current retention probabilities for items in the reservoir. 
In order to preserve unbiasedness, the
weight of any such item is normalized to $x/p_t(x)=\max\{x,z^*_t\}$.
Subsequent items in the aggregate that find their key already stored
are selected with probability $1$ and their sizes passed to \psh\
without any such initial normalization. 
Unbiasedness of the final estimate then follows from the chain rule for condition expectations
(see e.g \cite{W91}) since \psh\ provides an unbiased estimate of the
unbiased estimate produced by the \ash\ stage.
We note that \ash\ pre-sampler uses the \psh\ data structure to
determine whether a key is in storage. All key insertion and deletions
are handled by \psh\ component. We specify \papsh\ formally in
Algorithm~\ref{alg-psh-com} of Section~\ref{sec:structure}

\subsection{Trading Bias for MSE: Error Filtering} \label{sec:error}
Unbiased estimation of aggregates is effective for
larger aggregates since averaging over estimated contributions to the
aggregate reduces error. Smaller aggregates do not enjoy this
property, motivating supplementary approaches to reduce error. A
strawman approach is to count the number of estimates terms in the
aggregate, and use this value as a criterion to adjust or exclude
small aggregates. Another strawman approach filters based on
estimated variance, excluding aggregates with a high estimated
relative variance. The disadvantage of these approaches 
is that they require another counter. Instead, we are drawn to
find mechanisms to accomplish this goal that do not require extra
storage.

Our approach is quite simple: we ignore the contribution of the first item of every
newly instantiated aggregate to its estimate, although in
all other respects, sampling proceeds as before. Thus, while the renormalized
item weight does not contribute to the aggregate estimator $\hat X_k$,
the unnormalized item weight does contribute to $W_k$ used in
Theorem~\ref{thm:simple}. The resulting estimator is clearly biased since it
underestimates the true aggregate on average, but reduces  
as the experiments reported in Section~\ref{sec:eval} will
show.

\section{Algorithms and Implementation} \label{sec:structure}

\subsection{Algorithm Details}\label{sec-algorithm}

\begin{table}
{\small
\begin{tabular}{|l|l|l|}
\hline
Abbrev. & Description & Reference \\
\hline
\psh& Priority-Based Aggregation & Alg.~\ref{alg4} \\
\psh-EF& \psh\ w/ Error Filtering & Alg.~\ref{alg4} \\
\papsh& Priority-Based Adaptive Sample \& Hold &Alg.~\ref{alg-psh-com} \\
\papsh-EF & \papsh\ w/  Error Filtering & Alg.~\ref{alg-psh-com} \\
ASH & Adaptive Sample \& Hold & \cite{KME:sigmetrics05,Cohen:2015:SSF:2783258.2783279}\\
SH & Sample \& Hold (Non-Adaptive) & \cite{Estan2002}\\
\hline
\end{tabular}
}
\caption{Nomenclature for Algorithms}\label{tab:nomen}
\end{table}

The family of \psh\ algorithms using true weights is described in
Algorithm~\ref{alg4}. (Our nomenclature for the
Algorithms in given in
Table~\ref{tab:nomen}).
Pre-aggregation over consecutive items bearing
the same key (see Section~\ref{sec:preagg}) takes place in
lines~\ref{line:s1}--\ref{line:s2}. The pre-aggregates are passed to the
main loop in line ~\ref{line:s3}. In the main loop, deferred update
(Section~\ref{sec:deferred} takes place before aggregation to an
existing key in lines~\ref{alg:m5}--\ref{alg:m6}. Otherwise, a new key
entry is instantiated in lines~\ref{alg:m7}--\ref{alg:m8}. With error
filtering (Section~\ref{sec:error}), the first update of the
estimate is omitted at line~\ref{alg:erf}. When a new key arrives at the full
reservoir, selection of a key for discard takes place in
lines~\ref{alg:m9}-\ref{alg:m10}. In our implementation, we break this
step down further. The aggregates are maintained in a priority queue
implemented as a heap. An incoming new key is rejected if
its priority is less than the current minimum priority; see
Section~\ref{sec:costs}. After the
stream has been processed, remaining deferred
updates to the estimates occur in lines
\ref{alg:m4}--\ref{alg:m4:2}. This step could also be performed for any
or all aggregates in response to a query.
Algorithm~\ref{alg-psh-com} describes the modifications to the main
  loop for \papsh. A new pre-aggregate key is instantiated only if it
  passes the Sample and Hold admission test at line (\ref{alg:m11}).

\def\told{_{\mathrm{old}}}
\def\tnew{_{\mathrm{new}}}
\def\ttot{_{\mathrm{tot}}}
\begin{algorithm}
\setstretch{0.85}
\SetKwInOut{Input}{Input}
    \SetKwInOut{Output}{Output}
    \SetKwFunction{update}{update}
\SetKwFunction{mainloop}{mainloop}
\SetKwFunction{pshtrue}{PBA}
\SetKwFunction{mainlooptruew}{mainloop\_TrueW}
\SetKwProg{myproc}{Procedure}{}{}
  \caption{\psh: Priority-Based Aggregation w/ Optional
  Error Filtering}
  \label{alg4}
\Input{Stream of keyed weights $(k,x)$}
\Output{Estimated keyed weights $\{(k,a(k)): k\in K\}$}
\SetVline \dontprintsemicolon 
\BlankLine
\myproc{\pshtrue{$m$}}{
  $K=\emptyset$; $z^*=0$; $k\told = $ first key $k$ \label{line:s1}\\
  \While{(new keyed weight $(k\tnew,x\tnew)$)}{ 
    \If{($k\tnew = k\told$)} {
    $x\ttot \mathrel{+}= x\tnew$
  }
  \Else{
      \mainloop{$k\told,x\ttot$}\\
      $k\told = k\tnew$;  $x\ttot = x\tnew$ \label{line:s2}
    }
  }
  \mainloop{$k\tnew,x\ttot$}\\ \label{line:s3}
  \ForEach{$(k'\in K)$} { \label{alg:m4}
    \update{$k',z^*$} \\
  } \label{alg:m4:2}
end\\
}
\SetVline \dontprintsemicolon 
\hrule
\myproc{\mainloop{$k,x$}}{
  \If{($k\in K$)} {
    \update{$k,z^*$}\\ \label{alg:m5}
    $a(k)\mathrel{+}=x$; $w(k)\mathrel{+}=x$ \label{alg:m6}\\
    break
  }
  $K=K\cup\{k\}$; $w(k)=x$; $q(k)=1$\\ \label{alg:m7}
  $a(k)=x$; \tcp*{Omit if Error Filter} \label{alg:erf}
  generate $u(k)$ uniformly in $(0,1]$\\ \label{alg:m8}
  \If{($|K|\le m$)}{
    break
  }
  $k^*=\argmin_{k'\in K}\{w(k')/u(k')\}$\\ \label{alg:m9}
  $z^*=\max\{z^*, w(k^*)/u(k^*)\}$\\
  $K=K\setminus\{k^*\}$; \\ \label{alg:m10}
  Delete $a(k^*)$, $u(k^*)$, $q(k^*)$, $w(k^*)$\\
}
\hrule
\myproc{\update{$\tilde k, \tilde z$}}{
  $a(\tilde k) = a(\tilde k) * q(\tilde k)$\\
  $q(\tilde k)=\min\{q(\tilde k),w(\tilde k)/\tilde z\}$\\
  $a(\tilde k) = a(\tilde k) / q(\tilde k)$\\
}
\end{algorithm}

\begin{algorithm}
\setstretch{0.85}
	\SetKwInOut{Input}{Input}
	\SetKwInOut{Output}{Output}
	\caption{Priority-Based Adaptive Sample and Hold \papsh\ w/ Optional Error Filtering; mainloop only}
	\label{alg-psh-com}
	\SetKwFunction{update}{update}
	\SetKwFunction{mainloop}{mainloop}
	\SetKwProg{myproc}{Procedure}{}{}
	\SetVline \dontprintsemicolon 
	\BlankLine
	\myproc{\mainloop{$k,x$}}{
		\If{($k\in K$)} {
			\update{$k,z^*$}\\
			$a(k)\mathrel{+}=x$; $w(k)\mathrel{+}=x$\\
			break
		}
		Generate $r$ uniformly in $(0,1]$\\
		\If{$r<min(1,x/z^*)$}{ \label{alg:m11}
			$K=K\cup\{k\}$\\
			$a(k)=\max(x,z^*)$\tcp*{Omit if Error Filter}
			$w(k)=x$\\
			$q(k)=1$\\
			generate $u(k)$ uniformly in $(0,1]$\\
		}
		\If{($|K|\le m$)}{
			break
		}
		$k^*=\argmin_{k'\in K}\{w(k')/u(k')\}$\\
		$z^*=\max\{z^*, w(k^*)/u(k^*)\}$\\
		$K=K\setminus\{k^*\}$; \\
		Delete $a(k^*)$, $u(k^*)$, $q(k^*)$, $w(k^*)$\\
	}

\end{algorithm}

\begin{figure*}[htp!]
	\begin{multicols}{2}
		\includegraphics[width=\linewidth]{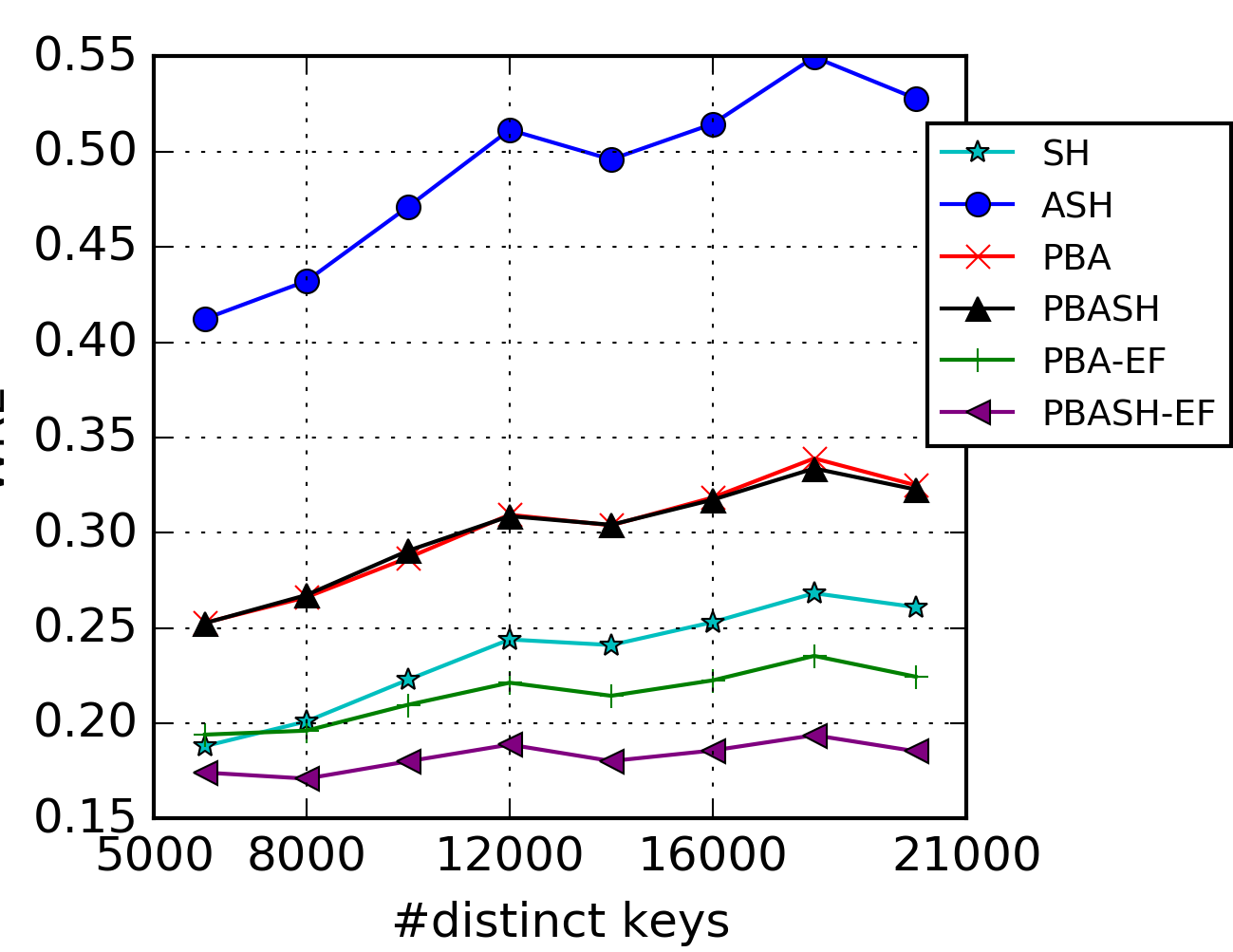}\par\caption{Weighted
                  relative error over all keys as a function of
                  distinct key count in reservoir size $m=1,000$.}
		\label{fig:bias_flow_error_for_byte_counts}
	\includegraphics[width=\linewidth]{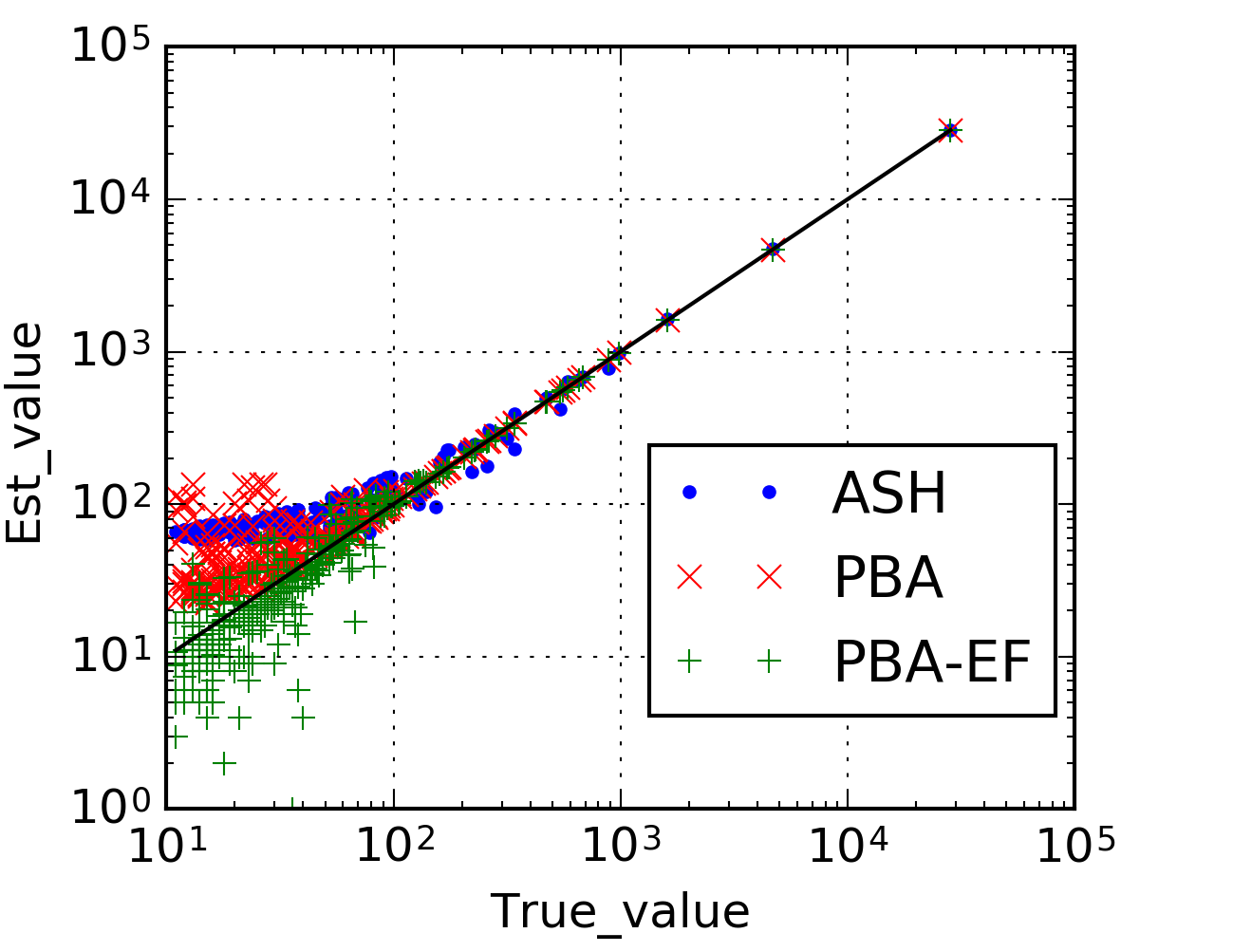}\par  \caption{Scatter
          plot of estimated vs. true aggregates for $10^4$ distinct keys sampled into reservoir size $m=500$.} 
		\label{fig:psh_scatter_plot_1000_10000}
	\end{multicols}
\end{figure*}

\subsection{Data Management \& Implementation Details}\label{sec:hashheap}

\def\hh{\textsc HashHeap}

In common with other stream aggregation schemes for (key, value) pairs,
\psh\ requires efficient access to the aggregate corresponding to the
incoming key $k$. Hash-tables provide an efficient means to achieve
this, with the hash $h(k)$ of the key $k$ referencing a location where
the aggregate, or in general its unbiased estimator is
maintained. 
\psh\ also maintains priorities as a priority queue. We implement this
as a heap. The question then arises of to efficiently combine
the heap and hash aspects of the aggregate store.

We manage this with
a combined structure called a \hh. 
This comprises two components. 
The first is a hash table that maps a key $k$ to a
pointer $\pi(k)$ into the second component. The second component is
a min-heap that maintains an entry
$(k,w,a,q)$ for each aggregate in storage, where $k$ is the key,
$u$ is the uniform random variable associated with $k$, $w$ the current incremented weight
since last admission, $a$ the current unbiased estimate, and $q$ the
current sampling probability. The heap is ordered by the priority
$r=w/u$ which is computed as required, with $u$ generated by hashing
on the key. The
heap is implemented in an array so that parent and child offsets can
be computed from the current offset of a key in the standard way.

\parab{Collision Resolution.}
In our design, keys are 
maintained in the heap, not in the hash. Collision identification
and resolution is performed by following a key $k$ to its position
$\pi(k)$ in the heap.  We illustrate for key insertion using linear
probing, which has been found to be extremely efficient for suitable
hash functions \cite{Patrascu:2012:PST:2220357.2220361}. Let $h$
denote the hash function. Suppose key $k$ is to be accessed. To find the
offset of key $k$ in the heap we probe the pointer hash table from
$h(k)$ until we find the pointer $\pi$ whose image in the hash table
is $k$. Probing to a vacant slot in the hash table indicates the key
is not in the heap. For insertion, the offset of the required location
in the heap is stored in the vacant slot in the hash table. Our approach is similar to one in \cite{Metwally:2005:ECF:2131560.2131596}, the difference being that in that work they key is maintained in the hash table, while each heap entry maintains a pointer back to is a corresponding hash entry. Our approach avoids storage for this second pointer, instead of computing it as needed from the key maintained in the heap.

\subsection{Computational and Storage Costs}\label{sec:costs}
\parab{Aggregation to an existing key is $O(1)$ average.} 
  All aggregation operations for a key $k$ are increasing its
  weight $w$ and hence also for its priority. Aggregation requires
  realignment of the heap, which is performed by bubbling
  down. i.e. swapping an element with its smallest priority child
  until it no longer has a larger priority than the child. The
  pointer offsets of the children are computed from the key $k$ as
  outlined above. The average cost for aggregation operation is $O(1)$. For simplicity, we
assume a perfectly balanced tree of depth $h$ and that the key to be
aggregated is uniformly distributed in the heap. Then the average
bubble down cost is no worse than $\sum_{\ell=0}^h 2^{\ell-h}(h-\ell)\le 2$.

\parab{Rejection of New Keys is $O(1)$ worst case.}
When an arriving item $(k,x)$ is not
present in reservoir, its priority is computed and compared with the
lowest priority item in the heap. Access to this item is $O(1)$. If
arriving item has lower priority it is discarded. The estimates of the
remaining items must be updated, but as established in Section~\ref{sec:deferred},
each update for a given key can be deferred until the next arrival
bearing that key.

\parab{Insertion/eviction for New Key is $O(\log m)$ worst case.} If the arriving item
has higher priority than the root item, the later is discarded, the new item inserted
at the root, then bubble down to its correct position in the
heap. This has worst case cost $O(\log m)$ for a reservoir of size $m$.

\parab{Retrieval is $O(1)$ per aggregate.} Any aggregate must undergo
a final deferred update prior to retrieval, incurring an $O(1)$ cost.

\parab{Storage Costs.} \textsl{Final Storage.} \psh, \prash\ and \ash\ all have the same final storage cost, requiring a (key, estimate) pair for all stored aggregated.
\textsl{Working Storage:} \psh\ and \prash\ are most costly for working storage,
requiring additional space per item for $q,w$ and the \hh\
pointer. The quasirandom number $u$ can be computed on demand by
hashing. These are maintained during stream aggregation, but discarded
at the end.

\begin{figure*}[htp!]
	\begin{multicols}{2}	
\includegraphics[width=\linewidth]{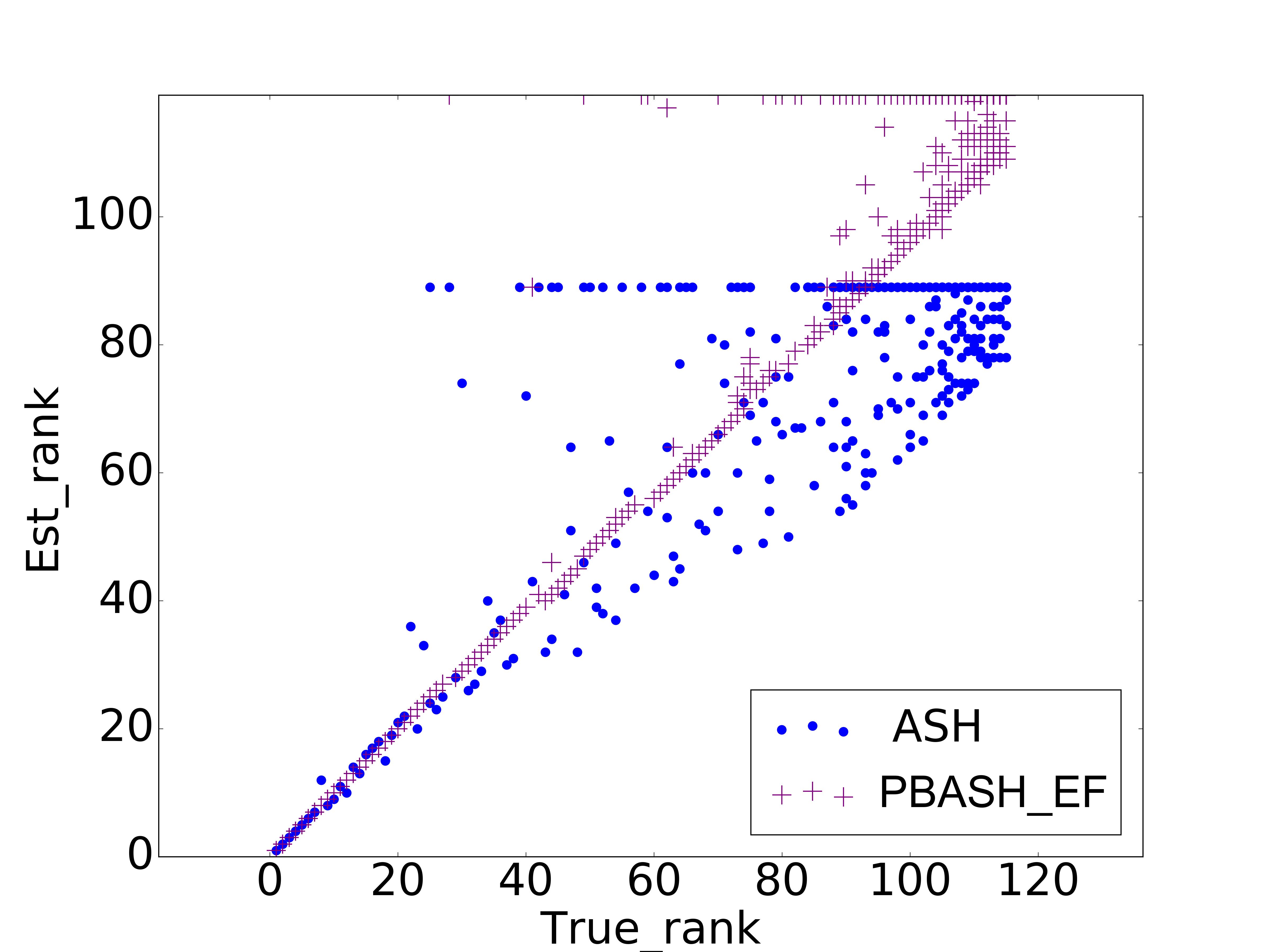}\par  \caption{Scatter of Estimated, Actual dense ranks, \prash\ and \ash. 5\% sampling; data as Figure~\ref{fig:psh_scatter_plot_1000_10000}.}
		\label{fig:ASH_PrSH_EF_rank}
\includegraphics[width=\linewidth]{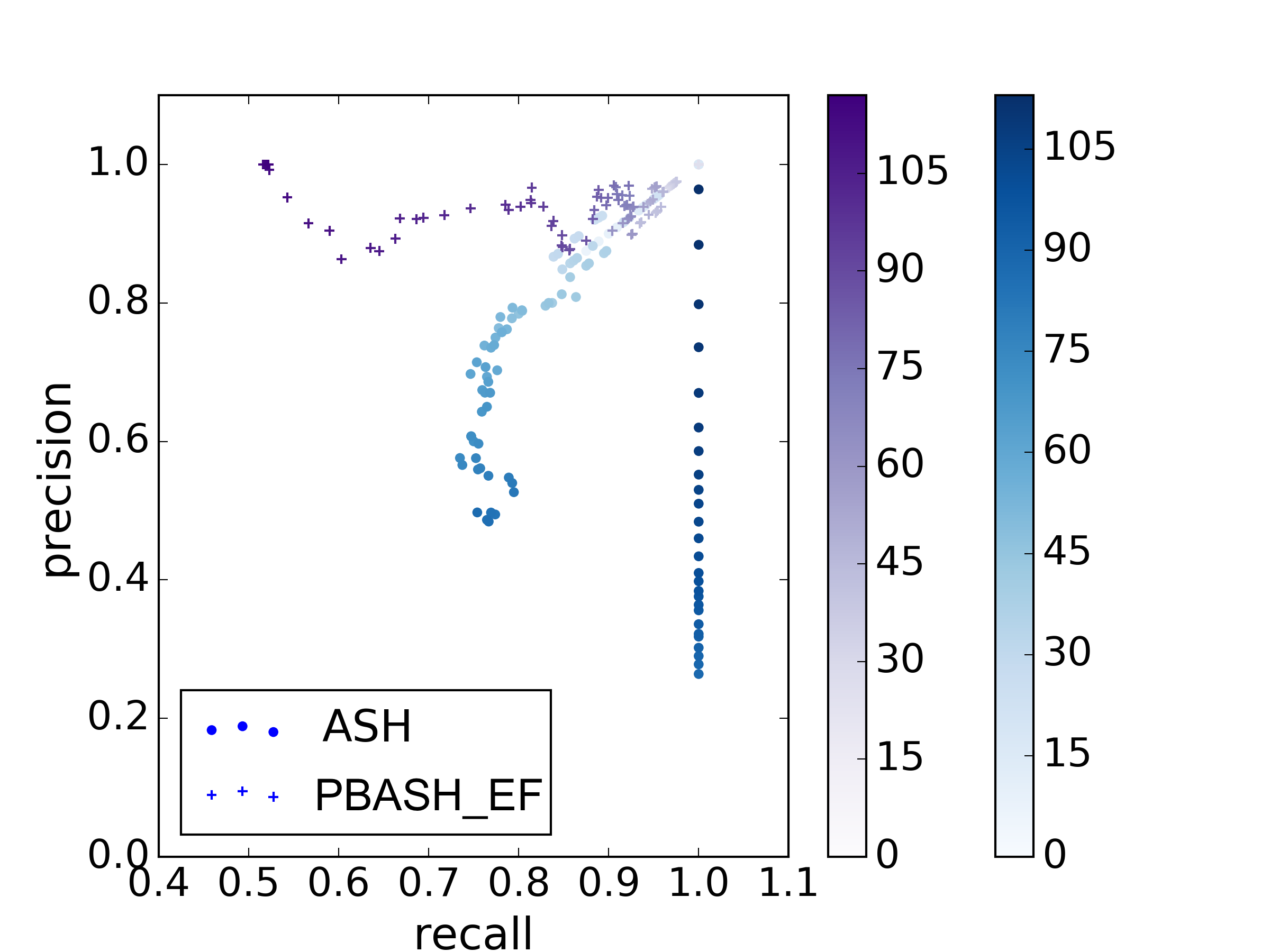}\par  \caption{Scatter of Prec($R$), Recall($R$) for dense ranks, rank $R$ on colormap. 5\% sampling; data as Figure~\ref{fig:psh_scatter_plot_1000_10000}.}
	\label{fig:ASH_PrSH_EF_colormap_maxRank_cb}
	\end{multicols}
\end{figure*}

\section{Evaluation}\label{sec:eval}

This section comprises a performance evaluation for \psh\ and \prash\ for accuracy and space and time complexity. We used both synthetic trace with features mimicking observed statistical behavior of network traffic, and real-word network traces from measured network denial of service attacks. These traces are chosen to represent dynamic network traffic, and serve to stress-test the summarization algorithms in their ability to adapt to dynamic conditions.
The evaluation represents measurement of network traffic over short time scales (at the time scale of seconds or shorter) that are if increasing interest for use in fine-scale traffic management in data center networks \cite{li2016flowradar}.

\subsection{Traces and Evaluation Metrics}

\parab{Trace Data and Platform.} The simulations ran on a 64-bit desktop equipped
with an Intel\textregistered\ Core\texttrademark\ i7-4790 Processor with 4 cores running at 3.6 GHz, each trial taking several seconds to tens of seconds. 
\begin{trivlist}
\item
\textsl{Trace 1: Synthetic Trace.} This trace was generated first by specifying a key set ranging in size from $6\times 10^3$ to $2\times 10^4$, and then for each key generating a set of unit weighted items whose number is drawn independently from a Pareto distribution with parameter $1.2$. The items are presented in random order. This trace is motivated by the observed heavy-tailed distribution of packets per flow aggregate in network traffic \cite{FRC98}.
\item \textsl{Trace 2: Network Trace with Distributed Denial of Service Attack (DDoS).} 
This trace is used to emulate the effect of network flooding with small packets
The traces is a 1-second CAIDA trace with $4.7\times 10^5$ packets and $62299$ distinct tuples (srcIP, dstIP, srcPort, dstPort, protocol) randomly mixed by 1-second DDoS traces~\cite{caida} 
with packet sending rate from $1.6\times 10^4$ to $6.0\times 10^6$ packets per second and distinct tuples from $6.4\times 10^3$ to $4.5\times 10^5$. The average size of one packet in the CAIDA trace is $495.5$ Bytes and that of the DDoS trace is $65.5$ Bytes.
\item[] \textsl{Trace 3: Dynamic Network Trace.}
The trace adds noise to a 15-second CAIDA trace. For each second, let the total byte volume be $V$, we generate a random probability $p\in(0,1)$, and $pV$ noise from another CAIDA trace is added to the original 1-second trace. 
\end{trivlist}

\begin{figure*}[htp!]
	\begin{multicols}{2}	
\includegraphics[width=\linewidth]{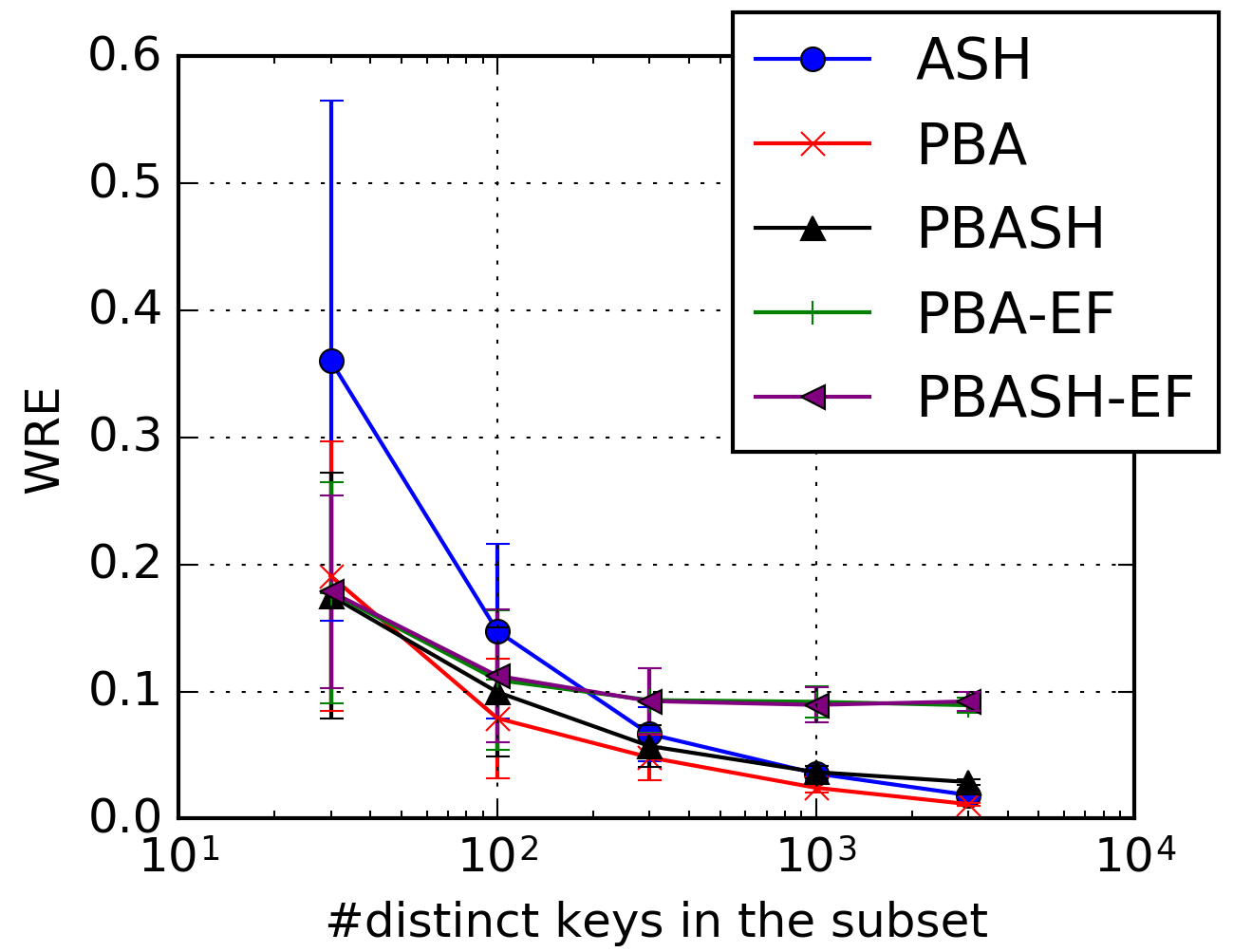}\par  \caption{WRE  as a function of subpopulation size
over 100 trials for $10^4$ distinct keys sampled into reservoir size $m=500$.}
		\label{fig:serror_df_csv}	
\includegraphics[width=\linewidth]{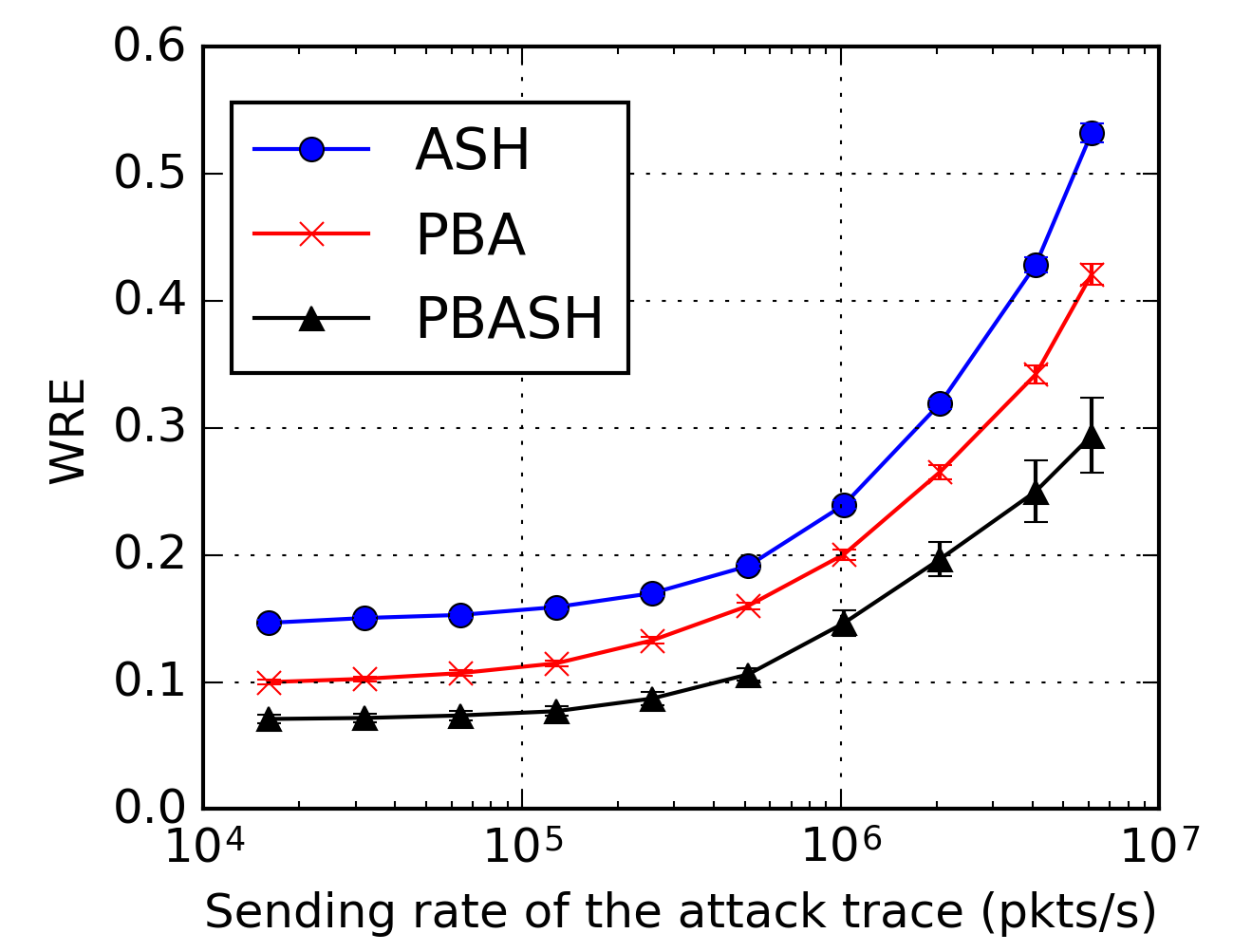}\par  \caption{WRE for mixed DDos traces at varying packet sending rates, and reservoir size $5,000$.}
		\label{fig:sending_rate_error_for_byte_counts}
	\end{multicols}
\end{figure*}

\parab{Evaluation Metrics.} The following metrics are measured against reservoir size, set as an independent variable, averaged over $100$ trials. For each trial, we randomize the order of the items in the traces. In addition, we randomly regenerate \textsl{Trace 1} for each trial.  
\begin{trivlist}
\item[] \textsl{Execution time:} This is the average time per packet over a trace
\item[] \textsl{Subpopulation Accuracy:} Our accuracy metric is the Weighted Relative Error (WRE), which we apply in two forms. The first is the average ${\sum_k|\hat X_k - X_k|}/{\sum_{k}X_k}$ where the sum runs over all distinct keys $k$. To evaluate accuracy for subpopulation queries we use a similar metric
$\sum_S |\hat X(S) - X(S)|/\sum_{S}X_S$ where $X(S)=\sum_{k\in S}x_k$ is the subset sum over a keyset $S$, and the sum runs over randomly chosen keysets $S\subset K$ of a given size $t$.
\item \textsl{Ranking Accuracy:} We compute accuracy for top-$R$ dense rank queries. In dense ranking, items with the same value receive the same rank, and ranks are consecutive. This avoids permutation noise of equal value; we also round estimates so as to reduce statistical noise. Let $\hat N(R)$ (respectively) and $N(R)$ denote the set of keys with true (respectively estimated) dense rank $\le R$. Then for a top-$R$ rank query, the precision and recall are
\be
\mbox{Prec}(R) = \frac{|N(R)\cap \hat N(R)|}{\hat N(R)}\mbox{ and }
\mbox{Rec}(R) = \frac{|N(R)\cap \hat N(R)|}{N(R)}
\ee
\end{trivlist}

\subsection{Accuracy Comparisons}

Figure~\ref{fig:bias_flow_error_for_byte_counts} illustrates error metrics for \psh, \prash, and \ash\ in a reservoir of size $m=1,000$ processing items from the synthetic Trace 1. The number of distinct keys varies from 6,000 to 20,000, representing a key sampling rate ranging from $17\%$ down to $5\%$,
WRE was reduced, relative to \ash, by about 40\%  for \psh\ and \prash,  by 53--57\%  for \psh-EF, and by 58--65\% for \prash-EF.
As shown, \prash\ and \prash-EF are able to achieve lower WRE than a best-case (non-adaptive) Sample and Hold (SH) in which the sampling rate is chosen so as minimize WRE.

To better understand the difference in error between \psh, \psh-EF and \ash, we drill down within an individual experiment. Figure~\ref{fig:psh_scatter_plot_1000_10000} is a scatter plot of estimated vs. true aggregate for the two methods for a synthetic trace containing $10^4$ distinct keys sampled into a reservoir of size $500$, i.e., a key sampling rate of $5\%$. The figure shows how \psh\ improves estimation accuracy for smaller weight keys, \ash\ having a larger additive error (note the logarithmic vertical axis).  As expected, \psh-EF\ further reduces the estimation error for small aggregates, typically underestimating the true value. 

\parab{Rank Estimation.} We evaluate rank estimation performance, focusing on algorithms involving Error Filtering since rankings should be less sensitive to bias than variability. Figure~\ref{fig:ASH_PrSH_EF_rank} shows a scatter plot of estimates vs. actual dense ranks at 5\% sampling for \prash-EF and \ash. Although both perform well for low ranks (larger aggregates), we observe increasing rank noise for \ash\ in mid to low ranks. The horizontal clusters in each case correspond to aggregates not sampled; there are noticeably more of these of lower true rank for \ash\ than \psh-EF. Figure~\ref{fig:ASH_PrSH_EF_colormap_maxRank_cb} shows precision and recall for top-$R$ rank queries. Precision is noticeably better for \prash-EF, particularly for middle ranks.

\parab{Subpopulation Weight Estimation.} Figure~\ref{fig:serror_df_csv} shows WRE for subpopulations over 100 random selected subpopulations as a function of subpopulation size. For small subpopulations up to size 100, \psh\ and derived methods provide up to about a 60\% reduction in WRE relative to \ash. The WREs of the unbiased methods (\psh, \prash, \ash) behave similarly for larger subpopulation sizes due to averaging, while the bias of the error filtering methods persist.

\parab{Network dynamics.} 
We study the effect in accuracy on an emulated DDoS attack with Trace 2. Figure~\ref{fig:sending_rate_error_for_byte_counts} shows the effect on WRE as the DDoS traffic rate increases, in a reservoir of size 5,000. The number of distinct keys increases in proportion to the attack traffic rate, with legitimate traffic representing a smaller proportion of the total.
\psh\ and \prash\ achieve lower error than \ash, even as errors for all methods increase, and \prash-EF (not shown) achieves 60\% reduction in error compared with \ash. Figure~\ref{fig:dynamic_error_for_byte_counts} shows a time series of WRE for the dynamic traffic of Trace 3, with samples taken over successive 250ms windows in a reservoir size 5,000. \psh\ and \prash\ have smaller fluctuations in WRE
in response to the dynamics
than \ash\, achieving similar reduction as before.

\begin{figure*}[htp!]
	\begin{multicols}{2}
\includegraphics[scale=0.71]{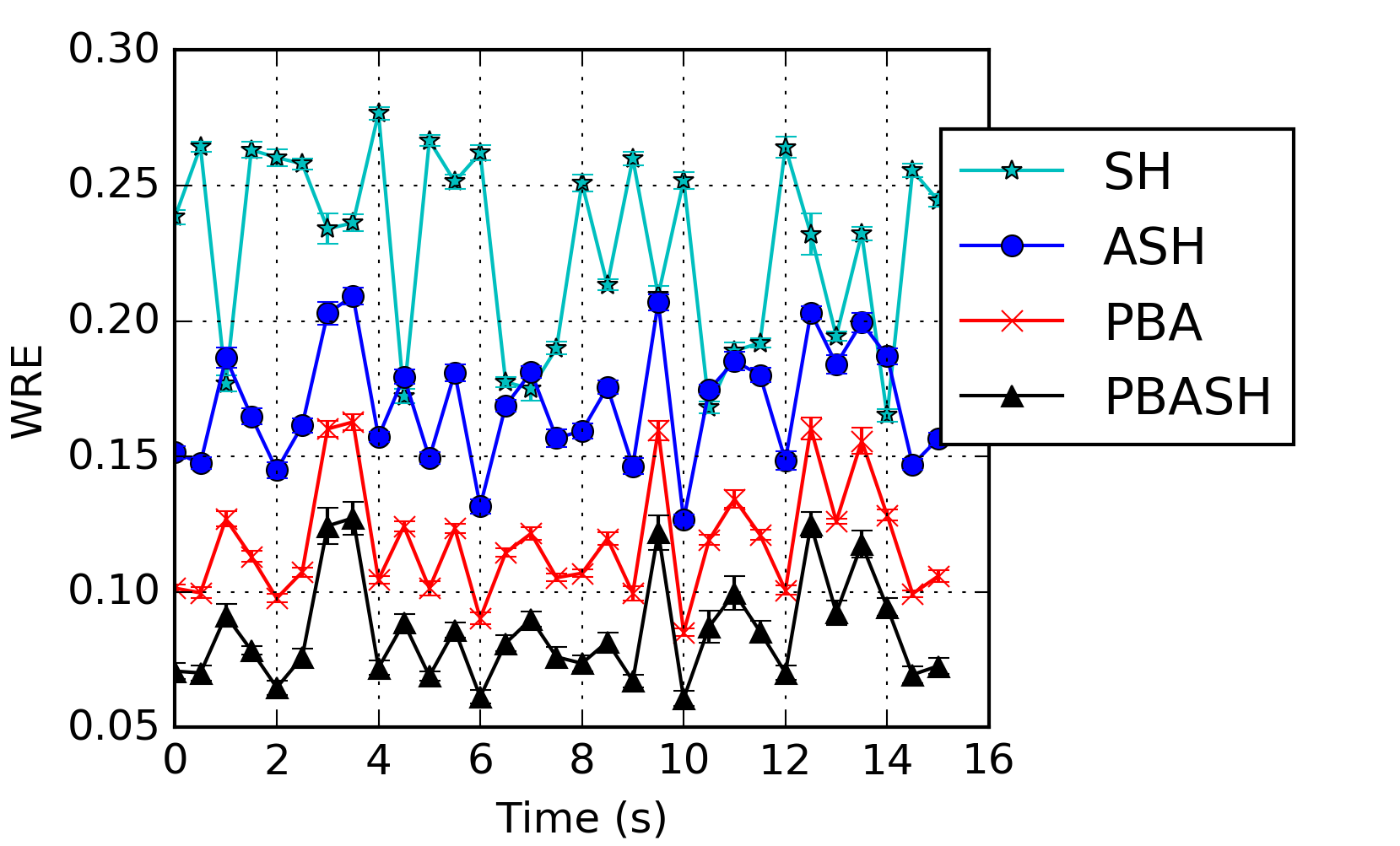}\caption{The impact of traffic dynamics by adding random noise when the reservoir size is $5,000$.}
	\label{fig:dynamic_error_for_byte_counts}
\includegraphics[width=0.9\linewidth]{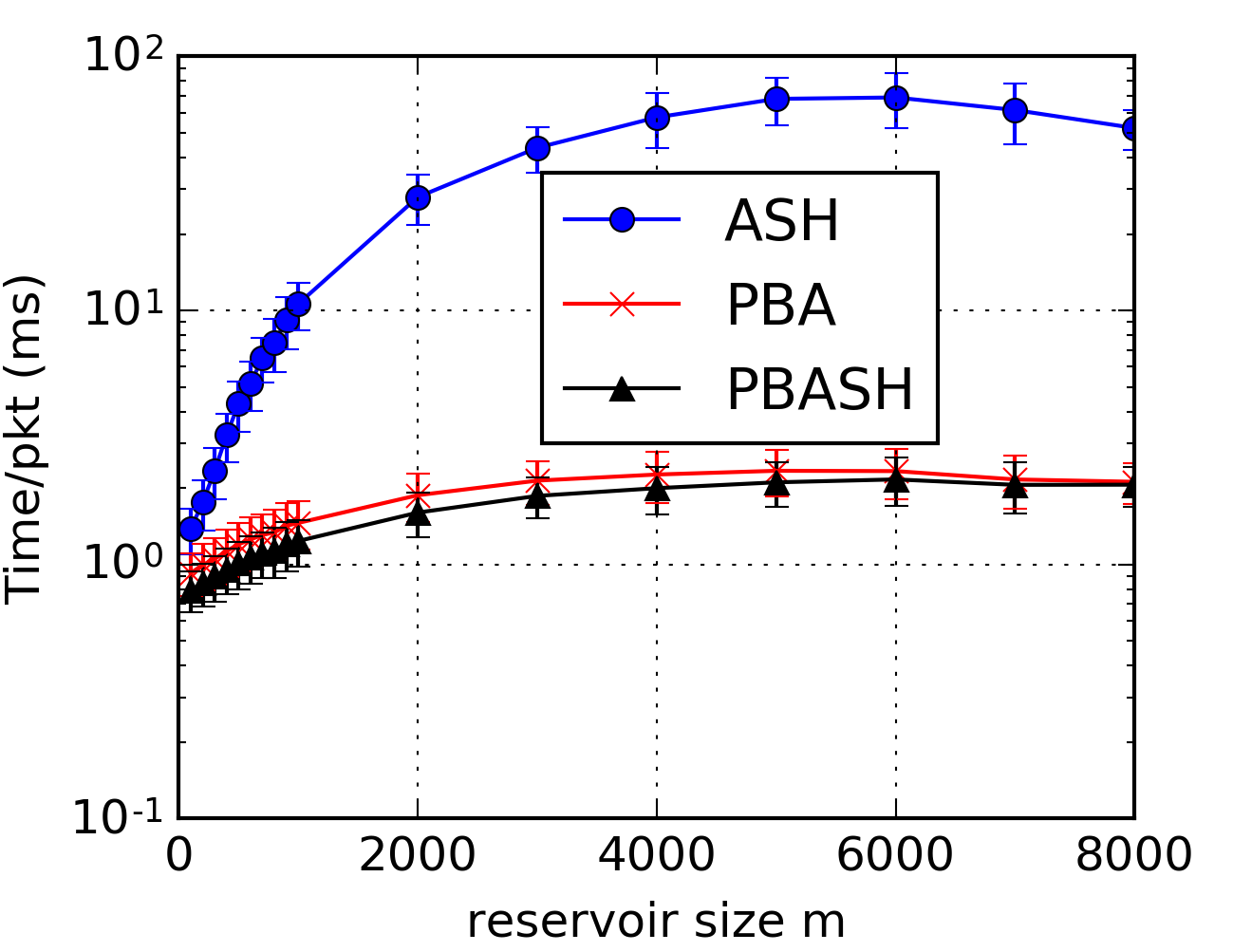}\par\caption{The time complexity compared to \ash\ with varying reservoir sizes and $10^4$ distinct keys.}
		\label{fig:countersize_time_complexity_for_byte_counts}
	\end{multicols}
\end{figure*}

\begin{figure}[htp!]
\includegraphics[width=0.9\linewidth]{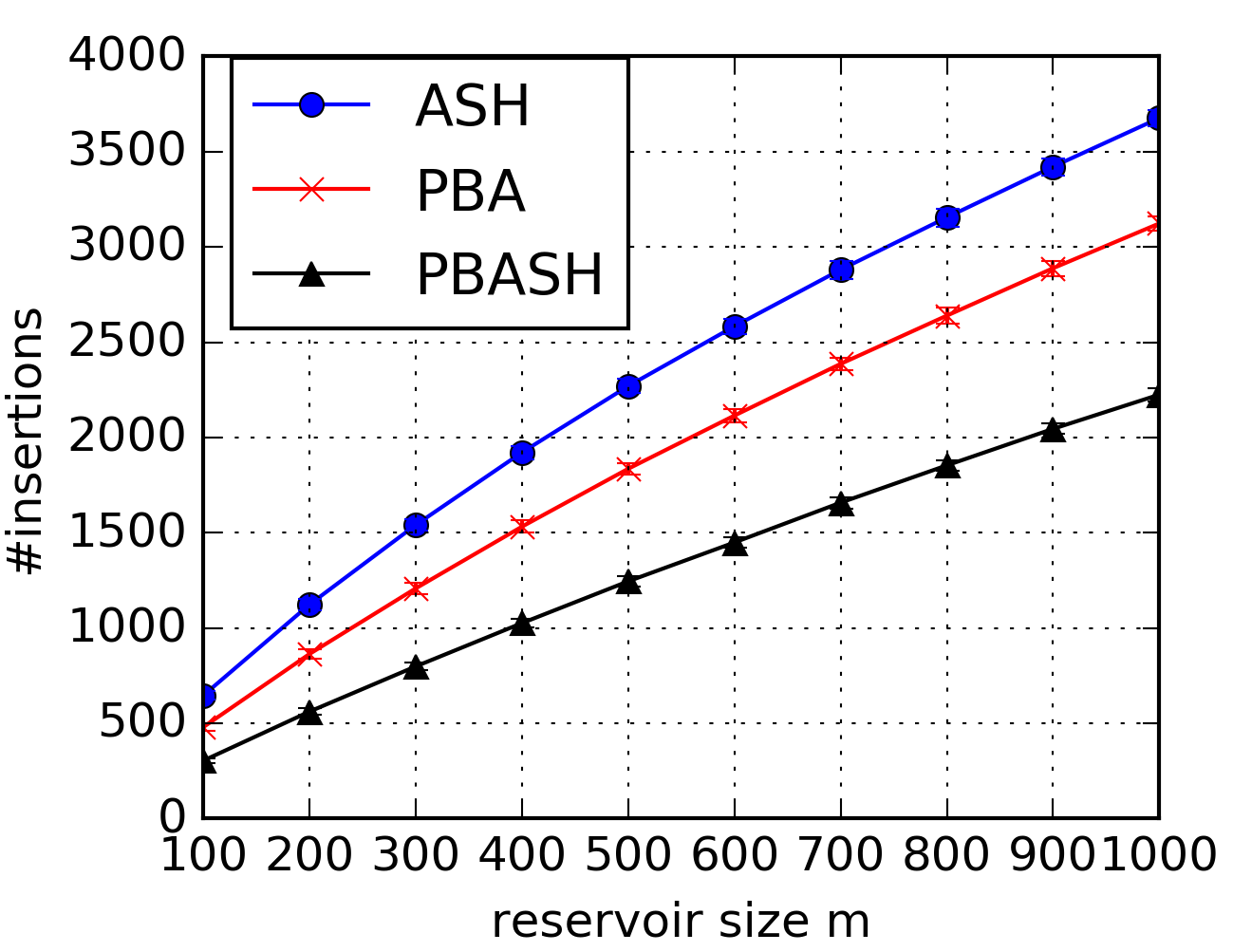}\par  \caption{The number of insertions when the reservoir size is from $100$ to $1,000$.} 
		\label{fig:insertions_countersize_time_complexity_for_byte_counts}		
\end{figure}

\subsection{Computational Complexity}
Figure~\ref{fig:countersize_time_complexity_for_byte_counts} shows the processing time per packet of \psh, \prash\ and \ash. No optimizations of \ash\ were used beyond the specification in \cite{Cohen:2015:SSF:2783258.2783279}.
With this proviso, the $O(m)$ cost for key eviction from reservoir size $m$ for \ash\ appears evident through the initial linear growth of the time per packet. The noticeably lower growth for \psh\ and \prash\ are expected due to its $O(\log m)$ time for inserting a new key after eviction of a current key.  Since insertion/eviction is the most costly part for all algorithms we display the experimental number of these for each algorithm in Figure~\ref{fig:insertions_countersize_time_complexity_for_byte_counts}. \prash\ has about half the insertions of \ash, another factor in its smaller per packet time. \prash\ also has a smaller number of insertions than \psh. This is to be expected, since the \prash\ pre-sampling stages causes fewer keys to be admitted to the reservoir.

\section{Conclusions}\label{sec:conclude}
Weighted sample-based algorithms are a flexible approach to stream
summarization, whose outputs can be readily utilized by downstream
applications for queries on ranks and subpopulations. This
paper provides a new set of algorithms, Priority-Based Aggregation and
its variants) of this type. \psh\ is designed around a single random
variable per key aggregate, allowing considerable speed-up in a fixed
cache, and it also improves accuracy for a given sample size, compared
with state-of-the-art methods.

\section{Proofs of the Theorems}\label{sec:proofs}

\begin{proof}[Proof of Theorem~\ref{thm:simple}] For each $k$ we proceed by induction
  on $t\ge s_k=\min\{s: k_s = k\}$ 
and establish that
\be\label{eq:unb:iter}
\E[\hat X_{k,t}|\hat X_{k,t-1},C]-X_{k,t}=\hat X_{k,t-1}-X_{k,t-1}\ee
for all members $C$ of a covering partition (i.e., a set of
  disjoint events whose union is identically true). Since $\hat
  X_{k,s_k-1}=X_{k,s_k-1}=0$ we conclude that $\E[\hat X_{k,t}]=X_{k,t}$.

For $s_k\le s \le s'$ let $A_k(s)=\{k\notin \hat K_{s-1}\}$ (note
$A_k(s_k)$ is identically true), let $B_k(s,s')$ denote the event $\{k\in\hat K_s\dots,\hat
  K_{s'}\}$, i.e., that $k$ is in sample at all times in $[s,s']$
. Then
  for each $t\ge s_k$ the collection of events formed by
  $\{A_k(s)B_k(s,t-1):\ s\in[s_k,t-1]\}$, and $A_k(t)$
is a covering partition. 

(i) \textsl{Conditioning on $A_k(t)$}  On $A_k(t)$, $k_t\ne k$
implies $\hat X_{k,t}=\hat X_{k,t-1} = 0 = X_{k,t}-X_{k,t-1}$. On the
other hand $k_t=k$ implies $t\in T^0$. Further conditioning on $z_{k,t}=\min_{j\in
  \hat K_{j,t-1}}W_{j,t-1}/u_j$ then (\ref{eq:iter}) tells us that
\be
\Pr[k\in \hat K_t | A_k(t),z_{k,t}]=\Pr[W_{k,t}/u_k \le z_{k,t}]=q_{k,t}
\ee
and hence regardless of $z_{k,t}$ we have
\be \E[\hat X_{k,t} | X_{k,t-1},A_k(t),z_{k,t}]=\hat X_{k,t-1} + X_{k,t}-X_{k,t-1}
\ee

(ii) \textsl{Conditioning on $A_k(s)B_k(s,t-1)$ any
  $s\in[s_k,t-1]$.}  Under this condition $k\in \hat K_{t-1}$ and if
furthermore $k_t\in \hat K_{t-1}$ then $t\notin T^0$ and the first
line in (\ref{eq:iter}) holds.
Suppose instead $k_t\notin K_{t-1}$ so that $t\in T^0$. Let
$\cZ_{k}(t,s)=\{z_{k,r}: r\in [s,t]^0\}$. Observing that
\be
\Pr[B_k(t,s)|A_k(s),\cZ_k(t,s)]=
\Pr[\cap_{r\in [s,t]^0}\{z_{k,r}\le \frac{W_{k,r}}{u_k}\}]=q_{k,t}
\nonumber
\ee
then
\bea
&&\kern -40pt \Pr[k\in\hat K_t| B_k(t-1,s)A_k(s),\cZ_k(t,s)]\\
&=&
\frac{\Pr[B_k(t,s) | A_k(s),\cZ_k(t,s)]}
{\Pr[B_k(t-1,s) | A_k(s),\cZ_k(t-1,s)]}
=\frac{q_{k,t}}{q_{k,\tau_t}} = Q_{k,t}\nonumber
\eea
and hence
\be\label{eq:cond:AB}
\E[\hat X_{k,t}| \hat X_{k,t-1},A_k(s),\cZ_k(t,s)] =\hat X_{k,t-1}
\ee
independently of the conditions on the LHS of (\ref{eq:cond:AB}). AS
noted above, $k\in \hat K_{t-1}$ on $B(t-1,s)$ hence
$X_{k,t}=X_{k,t-1}$ and we recover (\ref{eq:unb:iter}). Since we now
established (\ref{eq:unb:iter}) over all members $C$ of a covering
partition, the proof is complete.
\end{proof}

\begin{proof}[Proof of Theorem~\ref{thm:delay-alt}]
$t\in T$ means the arriving $k_t\ne d_t$ is admitted to the reservoir and hence
\be\label{eq:mod:iter}
z_t=\frac{W_{d_t,t}}{u_{d_t}}\ge \frac{W_{d_t,s}}{u_{d_t}}\ge z_s
\ee
for all $s\in [\tau_{d_t,t},t]^0$.
The first inequality follows because $W_{d_t,s}$ is nondecreasing on
the interval $[\tau_{d_t,t},t]^0$.  The second inequality follows
because the key
$d_t$ survives selection throughout $[\tau_{d_t,t},t]^0$ and hence its priority
cannot be lower than the threshold $z_s$ for any $s$ in that interval.
Since $d_t$ was admitted at $\tau_{d_t,t}$, then $d_{\tau_{d_t,t}}\ne d_t$ and hence we
apply the argument back recursively to the first sampling time
$m+1$. This establishes $z_t\ge z^*_t$ and hence $z_t=z_t^*$.

(ii) $i$ is admitted to $\hat K_t$ if $t\in T$ with $i=k_t\ne d_t$ and hence
by (i),
$q_{i,t}=\min\{1,W_{i,t}/z_t\}=\min\{1,w_{i,i}/z^*_t\}=q^*_{i,t}$. 
We
establish the general case by induction.
Assume $t\in T^0$ and $q_{i,s}=q^*_{i,s}$ for all $s\in[\tau_{i,t},\tau_t]^0$, and 
consider first the case that $z_{t}>z^*_{\tau_t}t$. Then $z^*_{t}=z_{t}$ hence
$q^*_{i,t}=q_{i,t}$.
If instead $z_{t}\le z^*_{\tau_t}$ then $z^*_{\tau_t}=z^*_{t}$ and
\be
\frac{W_{i,t}}{z_{t}}\ge \frac{W_{i,t}}{z^*_{t}} \ge
\frac{W_{i,\tau_t}}{z^*_{t}} = \frac{W_{i,\tau_t}}{z^*_{t}}
\ee
Thus we can replace $z_{t}$ by $z^*_{t}$ 
but
use of either leaves
the iterated value unchanged, since by the induction
hypothesis, both are greater than $q_{i,\tau_t}\le W_{i,\tau_t}/z^*_{i,\tau_t}$
\end{proof}

\bibliographystyle{ACM-Reference-Format}
\bibliography{paper,bibmaster,ahmed,fair,cycle,psh-paper,psh}


\begin{thebibliography}{00}


\ifx \showCODEN    \undefined \def \showCODEN     #1{\unskip}     \fi
\ifx \showDOI      \undefined \def \showDOI       #1{{\tt DOI:}\penalty0{#1}\ }
  \fi
\ifx \showISBNx    \undefined \def \showISBNx     #1{\unskip}     \fi
\ifx \showISBNxiii \undefined \def \showISBNxiii  #1{\unskip}     \fi
\ifx \showISSN     \undefined \def \showISSN      #1{\unskip}     \fi
\ifx \showLCCN     \undefined \def \showLCCN      #1{\unskip}     \fi
\ifx \shownote     \undefined \def \shownote      #1{#1}          \fi
\ifx \showarticletitle \undefined \def \showarticletitle #1{#1}   \fi
\ifx \showURL      \undefined \def \showURL       #1{#1}          \fi
\providecommand\bibfield[2]{#2}
\providecommand\bibinfo[2]{#2}
\providecommand\natexlab[1]{#1}
\providecommand\showeprint[2][]{arXiv:#2}

\bibitem[\protect\citeauthoryear{Ahmed, Duffield, Neville, and Kompella}{Ahmed
  et~al\mbox{.}}{2014}]%
        {ahmed2014gsh}
\bibfield{author}{\bibinfo{person}{Nesreen~K. Ahmed}, \bibinfo{person}{Nick~G.
  Duffield}, \bibinfo{person}{Jennifer Neville}, {and}
  \bibinfo{person}{Ramana~Rao Kompella}.} \bibinfo{year}{2014}\natexlab{}.
\newblock \showarticletitle{Graph Sample and Hold: A Framework for Big-Graph
  Analytics}. In \bibinfo{booktitle}{{\em ACM SIGKDD}}.
\newblock


\bibitem[\protect\citeauthoryear{Ahmed, Duffield, Willke, and Rossi}{Ahmed
  et~al\mbox{.}}{2017}]%
        {AhmedVLDB2017}
\bibfield{author}{\bibinfo{person}{Nesreen~K. Ahmed}, \bibinfo{person}{Nick~G.
  Duffield}, \bibinfo{person}{Theodore~L. Willke}, {and}
  \bibinfo{person}{Ryan~A. Rossi}.} \bibinfo{year}{2017}\natexlab{}.
\newblock \showarticletitle{On Sampling from Massive Graph Streams}.
\newblock \bibinfo{journal}{{\em PVLDB\/}} \bibinfo{volume}{10},
  \bibinfo{number}{11} (\bibinfo{date}{August} \bibinfo{year}{2017}).
\newblock


\bibitem[\protect\citeauthoryear{Alon, Matias, and Szegedy}{Alon
  et~al\mbox{.}}{1999}]%
        {AMS99}
\bibfield{author}{\bibinfo{person}{N. Alon}, \bibinfo{person}{Y. Matias}, {and}
  \bibinfo{person}{M. Szegedy}.} \bibinfo{year}{1999}\natexlab{}.
\newblock \showarticletitle{The Space Complexity of Approximating the Frequency
  Moments}.
\newblock \bibinfo{journal}{{\it J. Comput. System Sci.}} \bibinfo{volume}{58},
  \bibinfo{number}{1} (\bibinfo{year}{1999}), \bibinfo{pages}{137--147}.
\newblock


\bibitem[\protect\citeauthoryear{Andoni, Krauthgamer, and Onak}{Andoni
  et~al\mbox{.}}{2010}]%
        {andoni}
\bibfield{author}{\bibinfo{person}{Alexandr Andoni}, \bibinfo{person}{Robert
  Krauthgamer}, {and} \bibinfo{person}{Krzysztof Onak}.}
  \bibinfo{year}{2010}\natexlab{}.
\newblock \bibinfo{booktitle}{{\em Streaming Algorithms from Precision
  Sampling}}.
\newblock \bibinfo{type}{{T}echnical {R}eport} 1011.1263.
  \bibinfo{institution}{arXiv}.
\newblock


\bibitem[\protect\citeauthoryear{Andreyev}{Andreyev}{2014}]%
        {fb2014dcfabric}
\bibfield{author}{\bibinfo{person}{Alexey Andreyev}.}
  \bibinfo{year}{2014}\natexlab{}.
\newblock \bibinfo{title}{{Introducing data center fabric, the next-generation
  Facebook data center network}}.
\newblock   (\bibinfo{year}{2014}).
\newblock
\newblock
\shownote{\url{https://code.facebook.com/posts/360346274145943/introducing-data-center-fabric-the-next-generation-facebook-data-center-network/}.}


\bibitem[\protect\citeauthoryear{Braverman and Ostrovsky}{Braverman and
  Ostrovsky}{2010}]%
        {Braverman:2010:ZFL:1806689.1806729}
\bibfield{author}{\bibinfo{person}{Vladimir Braverman} {and}
  \bibinfo{person}{Rafail Ostrovsky}.} \bibinfo{year}{2010}\natexlab{}.
\newblock \showarticletitle{Zero-one Frequency Laws} {\em
  (\bibinfo{series}{STOC '10})}. \bibinfo{publisher}{ACM},
  \bibinfo{address}{New York, NY, USA}, \bibinfo{pages}{281--290}.
\newblock
\showISBNx{978-1-4503-0050-6}


\bibitem[\protect\citeauthoryear{CAIDA}{CAIDA}{2012}]%
        {caida}
\bibfield{author}{\bibinfo{person}{CAIDA}.} \bibinfo{year}{2012}\natexlab{}.
\newblock \bibinfo{title}{{CAIDA Anonymized Internet Traces}}.
\newblock   (\bibinfo{year}{2012}).
\newblock
\newblock
\shownote{\url{http://www.caida.org/data/passive/passive_2012_dataset.xml}.}


\bibitem[\protect\citeauthoryear{{Cisco Systems}}{{Cisco Systems}}{2012}]%
        {netflow}
\bibfield{author}{\bibinfo{person}{{Cisco Systems}}.}
  \bibinfo{year}{2012}\natexlab{}.
\newblock \bibinfo{title}{Introduction to Cisco IOS NetFlow - A Technical
  Overview}.
\newblock
  \bibinfo{howpublished}{\url{http://www.cisco.com/c/en/us/products/collateral/ios-nx-os-software/ios-netflow/prod_white_paper0900aecd80406232.html}}.
    (\bibinfo{year}{2012}).
\newblock


\bibitem[\protect\citeauthoryear{Cohen}{Cohen}{2015}]%
        {Cohen:2015:SSF:2783258.2783279}
\bibfield{author}{\bibinfo{person}{Edith Cohen}.}
  \bibinfo{year}{2015}\natexlab{}.
\newblock \showarticletitle{Stream Sampling for Frequency Cap Statistics}. In
  \bibinfo{booktitle}{{\em ACM SIGKDD}}. \bibinfo{address}{New York, NY, USA},
  \bibinfo{pages}{159--168}.
\newblock
\showISBNx{978-1-4503-3664-2}


\bibitem[\protect\citeauthoryear{Cohen, Cormode, and Duffield}{Cohen
  et~al\mbox{.}}{2012}]%
        {cohen2012don}
\bibfield{author}{\bibinfo{person}{Edith Cohen}, \bibinfo{person}{Graham
  Cormode}, {and} \bibinfo{person}{Nick Duffield}.}
  \bibinfo{year}{2012}\natexlab{}.
\newblock \showarticletitle{Don't let the negatives bring you down: sampling
  from streams of signed updates}.
\newblock \bibinfo{journal}{{\em In SIGMETRICS\/}} \bibinfo{volume}{40},
  \bibinfo{number}{1} (\bibinfo{year}{2012}), \bibinfo{pages}{343--354}.
\newblock


\bibitem[\protect\citeauthoryear{Cohen, Duffield, Kaplan, Lund, and
  Thorup}{Cohen et~al\mbox{.}}{2007}]%
        {Cohen:2007:AEA:1298306.1298344}
\bibfield{author}{\bibinfo{person}{Edith Cohen}, \bibinfo{person}{Nick
  Duffield}, \bibinfo{person}{Haim Kaplan}, \bibinfo{person}{Carsten Lund},
  {and} \bibinfo{person}{Mikkel Thorup}.} \bibinfo{year}{2007}\natexlab{}.
\newblock \showarticletitle{Algorithms and Estimators for Accurate
  Summarization of Internet Traffic}. In \bibinfo{booktitle}{{\em Proceedings
  of the 7th ACM SIGCOMM Conference on Internet Measurement}} {\em
  (\bibinfo{series}{IMC '07})}. \bibinfo{publisher}{ACM}, \bibinfo{address}{New
  York, NY, USA}, \bibinfo{pages}{265--278}.
\newblock
\showISBNx{978-1-59593-908-1}


\bibitem[\protect\citeauthoryear{Cohen and Kaplan}{Cohen and Kaplan}{2007}]%
        {bottomk07:ds}
\bibfield{author}{\bibinfo{person}{E. Cohen} {and} \bibinfo{person}{H.
  Kaplan}.} \bibinfo{year}{2007}\natexlab{}.
\newblock \showarticletitle{Summarizing data using Bottom-k sketches}. In
  \bibinfo{booktitle}{{\em Proceedings of the ACM PODC'07 Conference}}.
\newblock


\bibitem[\protect\citeauthoryear{Cohen and Kaplan}{Cohen and Kaplan}{2008}]%
        {bottomk:VLDB2008}
\bibfield{author}{\bibinfo{person}{E. Cohen} {and} \bibinfo{person}{H.
  Kaplan}.} \bibinfo{year}{2008}\natexlab{}.
\newblock \showarticletitle{Tighter estimation using bottom-k sketches}. In
  \bibinfo{booktitle}{{\em Proceedings of the 34th VLDB Conference}}.
\newblock
\showURL{%
\url{http://arxiv.org/abs/0802.3448}}


\bibitem[\protect\citeauthoryear{Cormode and Muthukrishnan}{Cormode and
  Muthukrishnan}{2004}]%
        {Cormode04animproved}
\bibfield{author}{\bibinfo{person}{Graham Cormode} {and} \bibinfo{person}{S.
  Muthukrishnan}.} \bibinfo{year}{2004}\natexlab{}.
\newblock \showarticletitle{An improved data stream summary: The Count-Min
  sketch and its applications}.
\newblock \bibinfo{journal}{{\em J. Algorithms\/}}  \bibinfo{volume}{55}
  (\bibinfo{year}{2004}), \bibinfo{pages}{29--38}.
\newblock


\bibitem[\protect\citeauthoryear{Cranor, Johnson, Spatscheck, and ladislav
  Shkapenyuk}{Cranor et~al\mbox{.}}{2003}]%
        {cranor2003}
\bibfield{author}{\bibinfo{person}{Chuck Cranor}, \bibinfo{person}{Theodore
  Johnson}, \bibinfo{person}{Oliver Spatscheck}, {and} \bibinfo{person}{V\
  ladislav Shkapenyuk}.} \bibinfo{year}{2003}\natexlab{}.
\newblock \showarticletitle{Gigascope: A Stream Database for Network
  Applications}. In \bibinfo{booktitle}{{\em Proc ACM SIGMOD}}.
\newblock


\bibitem[\protect\citeauthoryear{Duffield and Krishnamurthy}{Duffield and
  Krishnamurthy}{2016}]%
        {7852324}
\bibfield{author}{\bibinfo{person}{N. Duffield} {and} \bibinfo{person}{B.
  Krishnamurthy}.} \bibinfo{year}{2016}\natexlab{}.
\newblock \showarticletitle{Efficient sampling for better OSN data
  provisioning}. In \bibinfo{booktitle}{{\em 2016 54th Annual Allerton
  Conference on Communication, Control, and Computing (Allerton)}}.
  \bibinfo{pages}{861--868}.
\newblock


\bibitem[\protect\citeauthoryear{Duffield, Thorup, and Lund}{Duffield
  et~al\mbox{.}}{2007}]%
        {DLT:jacm07}
\bibfield{author}{\bibinfo{person}{N. Duffield}, \bibinfo{person}{M. Thorup},
  {and} \bibinfo{person}{C. Lund}.} \bibinfo{year}{2007}\natexlab{}.
\newblock \showarticletitle{Priority sampling for estimating arbitrary subset
  sums}.
\newblock \bibinfo{journal}{{\em J. Assoc. Comput. Mach.\/}}
  \bibinfo{volume}{54}, \bibinfo{number}{6} (\bibinfo{year}{2007}).
\newblock


\bibitem[\protect\citeauthoryear{Efraimidis and Spirakis}{Efraimidis and
  Spirakis}{2006}]%
        {ES:IPL2006}
\bibfield{author}{\bibinfo{person}{P.~S. Efraimidis} {and}
  \bibinfo{person}{P.~G. Spirakis}.} \bibinfo{year}{2006}\natexlab{}.
\newblock \showarticletitle{Weighted random sampling with a reservoir}.
\newblock \bibinfo{journal}{{\em Inf. Process. Lett.\/}} \bibinfo{volume}{97},
  \bibinfo{number}{5} (\bibinfo{year}{2006}), \bibinfo{pages}{181--185}.
\newblock


\bibitem[\protect\citeauthoryear{Estan and Varghese}{Estan and
  Varghese}{2002}]%
        {Estan2002}
\bibfield{author}{\bibinfo{person}{Cristian Estan} {and}
  \bibinfo{person}{George Varghese}.} \bibinfo{year}{2002}\natexlab{}.
\newblock \showarticletitle{New Directions in Traffic Measurement and
  Accounting}. In \bibinfo{booktitle}{{\em Proc. of SIGCOMM}}.
  \bibinfo{pages}{323--336}.
\newblock


\bibitem[\protect\citeauthoryear{Feldmann, Rexford, and Caceres}{Feldmann
  et~al\mbox{.}}{1998}]%
        {FRC98}
\bibfield{author}{\bibinfo{person}{Anja Feldmann}, \bibinfo{person}{Jennifer
  Rexford}, {and} \bibinfo{person}{Ramon Caceres}.}
  \bibinfo{year}{1998}\natexlab{}.
\newblock \showarticletitle{Efficient policies for carrying Web traffic over
  flow-switched networks}.
\newblock \bibinfo{journal}{{\em IEEE/ACM Transactions on Networking\/}}
  (\bibinfo{date}{December} \bibinfo{year}{1998}), \bibinfo{pages}{673--685}.
\newblock


\bibitem[\protect\citeauthoryear{Gibbons and Matias}{Gibbons and
  Matias}{1998}]%
        {GM:sigmod98}
\bibfield{author}{\bibinfo{person}{P. Gibbons} {and} \bibinfo{person}{Y.
  Matias}.} \bibinfo{year}{1998}\natexlab{}.
\newblock \showarticletitle{New sampling-based summary statistics for improving
  approximate query answers}. In \bibinfo{booktitle}{{\em SIGMOD}}. ACM.
\newblock


\bibitem[\protect\citeauthoryear{Gilbert, Guha, Indyk, Kotidis, Muthukrishnan,
  and Strauss}{Gilbert et~al\mbox{.}}{2002}]%
        {Gilbert:2002:FSA:509907.509966}
\bibfield{author}{\bibinfo{person}{Anna~C. Gilbert}, \bibinfo{person}{Sudipto
  Guha}, \bibinfo{person}{Piotr Indyk}, \bibinfo{person}{Yannis Kotidis},
  \bibinfo{person}{S. Muthukrishnan}, {and} \bibinfo{person}{Martin~J.
  Strauss}.} \bibinfo{year}{2002}\natexlab{}.
\newblock \showarticletitle{Fast, Small-space Algorithms for Approximate
  Histogram Maintenance} {\em (\bibinfo{series}{STOC '02})}.
  \bibinfo{publisher}{ACM}, \bibinfo{address}{New York, NY, USA},
  \bibinfo{pages}{389--398}.
\newblock
\showISBNx{1-58113-495-9}


\bibitem[\protect\citeauthoryear{H{\'a}jek}{H{\'a}jek}{1960}]%
        {hajek:1960}
\bibfield{author}{\bibinfo{person}{J. H{\'a}jek}.}
  \bibinfo{year}{1960}\natexlab{}.
\newblock \showarticletitle{Limiting Distributions in Simple Random Sampling
  from a Finite Population}.
\newblock \bibinfo{journal}{{\em Publications of Mathematical Institute of
  Hungarian Academy of Sciences, Series A\/}}  \bibinfo{volume}{5}
  (\bibinfo{year}{1960}), \bibinfo{pages}{361--374}.
\newblock


\bibitem[\protect\citeauthoryear{Indyk}{Indyk}{2000}]%
        {stable-sk}
\bibfield{author}{\bibinfo{person}{P. Indyk}.} \bibinfo{year}{2000}\natexlab{}.
\newblock \showarticletitle{Stable Distributions, Pseudorandom Generators,
  Embeddings and Data Stream Computation}. In \bibinfo{booktitle}{{\em Proc. of
  the 41st Symposium on Foundations of Computer Science}}.
\newblock


\bibitem[\protect\citeauthoryear{Johnson, Lindenstrauss, and
  Schechtman}{Johnson et~al\mbox{.}}{1986}]%
        {Johnson1986}
\bibfield{author}{\bibinfo{person}{William~B. Johnson}, \bibinfo{person}{Joram
  Lindenstrauss}, {and} \bibinfo{person}{Gideon Schechtman}.}
  \bibinfo{year}{1986}\natexlab{}.
\newblock \showarticletitle{Extensions of lipschitz maps into Banach spaces}.
\newblock \bibinfo{journal}{{\em Israel Journal of Mathematics\/}}
  \bibinfo{volume}{54}, \bibinfo{number}{2} (\bibinfo{year}{1986}),
  \bibinfo{pages}{129--138}.
\newblock
\showISSN{1565-8511}
\showDOI{%
\url{http://dx.doi.org/10.1007/BF02764938}}


\bibitem[\protect\citeauthoryear{Jowhari, Saglam, and Tardos}{Jowhari
  et~al\mbox{.}}{2011}]%
        {Jowhari:Saglam:Tardos:11}
\bibfield{author}{\bibinfo{person}{Hossein Jowhari}, \bibinfo{person}{Mert
  Saglam}, {and} \bibinfo{person}{G{\'a}bor Tardos}.}
  \bibinfo{year}{2011}\natexlab{}.
\newblock \showarticletitle{Tight bounds for {Lp} samplers, finding duplicates
  in streams, and related problems}. In \bibinfo{booktitle}{{\em PODS}}.
  \bibinfo{pages}{49--58}.
\newblock


\bibitem[\protect\citeauthoryear{Keys, Moore, and Estan}{Keys
  et~al\mbox{.}}{2005}]%
        {KME:sigmetrics05}
\bibfield{author}{\bibinfo{person}{K. Keys}, \bibinfo{person}{D. Moore}, {and}
  \bibinfo{person}{C. Estan}.} \bibinfo{year}{2005}\natexlab{}.
\newblock \showarticletitle{A robust system for accurate real-time summaries of
  internet traffic}.
\newblock \bibinfo{journal}{{\em ACM SIGMETRICS Performance Evaluation
  Review\/}}  \bibinfo{volume}{33} (\bibinfo{year}{2005}).
\newblock


\bibitem[\protect\citeauthoryear{Li, Miao, Kim, and Yu}{Li
  et~al\mbox{.}}{2016}]%
        {li2016flowradar}
\bibfield{author}{\bibinfo{person}{Yuliang Li}, \bibinfo{person}{Rui Miao},
  \bibinfo{person}{Changhoon Kim}, {and} \bibinfo{person}{Minlan Yu}.}
  \bibinfo{year}{2016}\natexlab{}.
\newblock \showarticletitle{FlowRadar: a better NetFlow for data centers}. In
  \bibinfo{booktitle}{{\em 13th USENIX Symposium on Networked Systems Design
  and Implementation (NSDI 16)}}. USENIX Association,
  \bibinfo{pages}{311--324}.
\newblock


\bibitem[\protect\citeauthoryear{Liu, Golab, Golab, Ilyas, and Jin}{Liu
  et~al\mbox{.}}{2016a}]%
        {Liu:2016:SMD:3015779.3004295}
\bibfield{author}{\bibinfo{person}{Xiufeng Liu}, \bibinfo{person}{Lukasz
  Golab}, \bibinfo{person}{Wojciech Golab}, \bibinfo{person}{Ihab~F. Ilyas},
  {and} \bibinfo{person}{Shichao Jin}.} \bibinfo{year}{2016}\natexlab{a}.
\newblock \showarticletitle{Smart Meter Data Analytics: Systems, Algorithms,
  and Benchmarking}.
\newblock \bibinfo{journal}{{\em ACM Trans. Database Syst.\/}}
  \bibinfo{volume}{42}, \bibinfo{number}{1}, Article \bibinfo{articleno}{2}
  (\bibinfo{date}{November} \bibinfo{year}{2016}),
  \bibinfo{numpages}{39}~pages.
\newblock
\showISSN{0362-5915}


\bibitem[\protect\citeauthoryear{Liu, Manousis, Vorsanger, Sekar, and
  Braverman}{Liu et~al\mbox{.}}{2016b}]%
        {Liu:2016:OSR:2934872.2934906}
\bibfield{author}{\bibinfo{person}{Zaoxing Liu}, \bibinfo{person}{Antonis
  Manousis}, \bibinfo{person}{Gregory Vorsanger}, \bibinfo{person}{Vyas Sekar},
  {and} \bibinfo{person}{Vladimir Braverman}.}
  \bibinfo{year}{2016}\natexlab{b}.
\newblock \showarticletitle{One Sketch to Rule Them All: Rethinking Network
  Flow Monitoring with UnivMon}. In \bibinfo{booktitle}{{\em Proceedings of the
  2016 Conference on ACM SIGCOMM 2016 Conference}} {\em
  (\bibinfo{series}{SIGCOMM '16})}. \bibinfo{publisher}{ACM},
  \bibinfo{address}{New York, NY, USA}, \bibinfo{pages}{101--114}.
\newblock
\showISBNx{978-1-4503-4193-6}


\bibitem[\protect\citeauthoryear{Lo, Frankowski, and Leskovec}{Lo
  et~al\mbox{.}}{2016}]%
        {Lo:2016:UBL:2939672.2939729}
\bibfield{author}{\bibinfo{person}{Caroline Lo}, \bibinfo{person}{Dan
  Frankowski}, {and} \bibinfo{person}{Jure Leskovec}.}
  \bibinfo{year}{2016}\natexlab{}.
\newblock \showarticletitle{Understanding Behaviors That Lead to Purchasing: A
  Case Study of Pinterest}. In \bibinfo{booktitle}{{\em Proceedings of the 22Nd
  ACM SIGKDD International Conference on Knowledge Discovery and Data Mining}}
  {\em (\bibinfo{series}{KDD '16})}. \bibinfo{publisher}{ACM},
  \bibinfo{address}{New York, NY, USA}, \bibinfo{pages}{531--540}.
\newblock
\showISBNx{978-1-4503-4232-2}


\bibitem[\protect\citeauthoryear{Metwally, Agrawal, and El~Abbadi}{Metwally
  et~al\mbox{.}}{2005}]%
        {Metwally:2005:ECF:2131560.2131596}
\bibfield{author}{\bibinfo{person}{Ahmed Metwally}, \bibinfo{person}{Divyakant
  Agrawal}, {and} \bibinfo{person}{Amr El~Abbadi}.}
  \bibinfo{year}{2005}\natexlab{}.
\newblock \showarticletitle{Efficient Computation of Frequent and Top-k
  Elements in Data Streams}. In \bibinfo{booktitle}{{\em Proceedings of the
  10th International Conference on Database Theory}} {\em
  (\bibinfo{series}{ICDT'05})}. \bibinfo{publisher}{Springer-Verlag},
  \bibinfo{address}{Berlin, Heidelberg}, \bibinfo{pages}{398--412}.
\newblock
\showISBNx{3-540-24288-0, 978-3-540-24288-8}


\bibitem[\protect\citeauthoryear{Monemizadeh and Woodruff}{Monemizadeh and
  Woodruff}{2010}]%
        {MoWo:SODA2010}
\bibfield{author}{\bibinfo{person}{M. Monemizadeh} {and} \bibinfo{person}{D.~P.
  Woodruff}.} \bibinfo{year}{2010}\natexlab{}.
\newblock \showarticletitle{1-Pass Relative-Error L$_{\mbox{p}}$-Sampling with
  Applications}. In \bibinfo{booktitle}{{\em Proc. 21st ACM-SIAM Symposium on
  Discrete Algorithms}}. ACM-SIAM.
\newblock


\bibitem[\protect\citeauthoryear{Moshref, Yu, Govindan, and Vahdat}{Moshref
  et~al\mbox{.}}{2015}]%
        {Moshref:2015:SSR:2716281.2836099}
\bibfield{author}{\bibinfo{person}{Masoud Moshref}, \bibinfo{person}{Minlan
  Yu}, \bibinfo{person}{Ramesh Govindan}, {and} \bibinfo{person}{Amin Vahdat}.}
  \bibinfo{year}{2015}\natexlab{}.
\newblock \showarticletitle{SCREAM: Sketch Resource Allocation for
  Software-defined Measurement}. In \bibinfo{booktitle}{{\em Proceedings of the
  11th ACM Conference on Emerging Networking Experiments and Technologies}}
  {\em (\bibinfo{series}{CoNEXT '15})}. \bibinfo{publisher}{ACM},
  \bibinfo{address}{New York, NY, USA}, Article \bibinfo{articleno}{14},
  \bibinfo{numpages}{13}~pages.
\newblock
\showISBNx{978-1-4503-3412-9}


\bibitem[\protect\citeauthoryear{P\v{a}tra\c{s}cu and Thorup}{P\v{a}tra\c{s}cu
  and Thorup}{2012}]%
        {Patrascu:2012:PST:2220357.2220361}
\bibfield{author}{\bibinfo{person}{Mihai P\v{a}tra\c{s}cu} {and}
  \bibinfo{person}{Mikkel Thorup}.} \bibinfo{year}{2012}\natexlab{}.
\newblock \showarticletitle{The Power of Simple Tabulation Hashing}.
\newblock \bibinfo{journal}{{\em J. ACM\/}} \bibinfo{volume}{59},
  \bibinfo{number}{3}, Article \bibinfo{articleno}{14} (\bibinfo{date}{June}
  \bibinfo{year}{2012}), \bibinfo{numpages}{50}~pages.
\newblock
\showISSN{0004-5411}


\bibitem[\protect\citeauthoryear{Ros{\'e}n}{Ros{\'e}n}{1972}]%
        {Rosen1972:successive}
\bibfield{author}{\bibinfo{person}{B. Ros{\'e}n}.}
  \bibinfo{year}{1972}\natexlab{}.
\newblock \showarticletitle{Asymptotic Theory for Successive Sampling with
  Varying Probabilities Without Replacement, {I}}.
\newblock \bibinfo{journal}{{\em The Annals of Mathematical Statistics\/}}
  \bibinfo{volume}{43}, \bibinfo{number}{2} (\bibinfo{year}{1972}),
  \bibinfo{pages}{373--397}.
\newblock
\showURL{%
\url{http://www.jstor.org/stable/2239977}}


\bibitem[\protect\citeauthoryear{Ros{\'e}n}{Ros{\'e}n}{1997}]%
        {rosen1997a}
\bibfield{author}{\bibinfo{person}{B. Ros{\'e}n}.}
  \bibinfo{year}{1997}\natexlab{}.
\newblock \showarticletitle{Asymptotic theory for order sampling}.
\newblock \bibinfo{journal}{{\em J. Statistical Planning and Inference\/}}
  \bibinfo{volume}{62}, \bibinfo{number}{2} (\bibinfo{year}{1997}),
  \bibinfo{pages}{135--158}.
\newblock


\bibitem[\protect\citeauthoryear{Roy, Zeng, Bagga, Porter, and Snoeren}{Roy
  et~al\mbox{.}}{2015}]%
        {roy2015inside}
\bibfield{author}{\bibinfo{person}{Arjun Roy}, \bibinfo{person}{Hongyi Zeng},
  \bibinfo{person}{Jasmeet Bagga}, \bibinfo{person}{George Porter}, {and}
  \bibinfo{person}{Alex~C Snoeren}.} \bibinfo{year}{2015}\natexlab{}.
\newblock \showarticletitle{Inside the social network's (datacenter) network}.
  In \bibinfo{booktitle}{{\em ACM SIGCOMM Computer Communication Review}},
  Vol.~\bibinfo{volume}{45}. ACM, \bibinfo{pages}{123--137}.
\newblock


\bibitem[\protect\citeauthoryear{Vitter}{Vitter}{1985}]%
        {Vitter:85}
\bibfield{author}{\bibinfo{person}{J. Vitter}.}
  \bibinfo{year}{1985}\natexlab{}.
\newblock \showarticletitle{Random sampling with a reservoir}.
\newblock \bibinfo{journal}{{\em ACM Trans. Math. Softw.\/}}
  \bibinfo{volume}{11} (\bibinfo{year}{1985}).
\newblock
Issue 1.


\bibitem[\protect\citeauthoryear{Williams}{Williams}{1991}]%
        {W91}
\bibfield{author}{\bibinfo{person}{David Williams}.}
  \bibinfo{year}{1991}\natexlab{}.
\newblock \bibinfo{booktitle}{{\em Probability with Martingales}}.
\newblock \bibinfo{publisher}{Cambridge University Press}.
\newblock


\bibitem[\protect\citeauthoryear{Yu, Jose, and Miao}{Yu et~al\mbox{.}}{2013}]%
        {Yu:2013:SDT:2482626.2482631}
\bibfield{author}{\bibinfo{person}{Minlan Yu}, \bibinfo{person}{Lavanya Jose},
  {and} \bibinfo{person}{Rui Miao}.} \bibinfo{year}{2013}\natexlab{}.
\newblock \showarticletitle{Software Defined Traffic Measurement with
  OpenSketch}. In \bibinfo{booktitle}{{\em Proceedings of the 10th USENIX
  Conference on Networked Systems Design and Implementation}} {\em
  (\bibinfo{series}{nsdi'13})}. \bibinfo{publisher}{USENIX Association},
  \bibinfo{address}{Berkeley, CA, USA}, \bibinfo{pages}{29--42}.
\newblock
\showURL{%
\url{http://dl.acm.org/citation.cfm?id=2482626.2482631}}


\end{thebibliography}

\end{document}